\documentclass[smallextended,natbib,runningheads]{svjour3}
\journalname{Annals of the Institute of Statistical Mathematics}
\smartqed  
\usepackage{graphicx}
\usepackage{amsmath}
\usepackage{amssymb}
\usepackage{dsfont}
\usepackage{url}
\usepackage{color}
\usepackage{natbib}
\usepackage{amssymb,amsmath}

\begin{document}

\title{On the usage of randomized p-values in the Schweder--Spj\o tvoll estimator}


\titlerunning{Randomized p-values in the Schweder--Spj\o tvoll estimator}        

\author{Anh-Tuan Hoang         \and
        Thorsten Dickhaus 
}


\institute{A. Hoang \at
              Institute for Statistics, University of Bremen, D-28344 Bremen, Germany\\
							Tel.: +49 421 218-63652\\
              \email{anhtuan.hoang@uni-bremen.de}           
           \and
           T. Dickhaus \at
              Institute for Statistics, University of Bremen, D-28344 Bremen, Germany \\
           Tel.: +49 421 218-63651  \\
					\email{dickhaus@uni-bremen.de}
}

\date{\today}

\maketitle

\begin{abstract}
We are concerned with multiple test problems with composite null hypotheses and the estimation of the proportion $\pi_{0}$ of true null hypotheses. The Schweder-Spj\o tvoll estimator $\hat{\pi}_0$ utilizes marginal $p$-values and only works properly if the $p$-values that correspond to the true null hypotheses are uniformly distributed on $[0,1]$ ($\mathrm{Uni}[0,1]$-distributed). In the case of composite null hypotheses, marginal $p$-values are usually computed under least favorable parameter configurations (LFCs).  Thus, they are stochastically larger than $\mathrm{Uni}[0,1]$ under non-LFCs in the null hypotheses. When using these LFC-based $p$-values, $\hat{\pi}_0$ tends to overestimate $\pi_{0}$. We introduce a new way of randomizing $p$-values that depends on a tuning parameter $c\in[0,1]$, such that $c=0$ and $c=1$ lead to $\mathrm{Uni}[0,1]$-distributed $p$-values, which are independent of the data, and to the original LFC-based $p$-values, respectively. For a 
certain value $c=c^{\star}$ the bias of $\hat{\pi}_0$ is minimized when using our randomized $p$-values. This often also entails a smaller mean squared error of the estimator as compared to the usage of the LFC-based $p$-values. We analyze these points theoretically, and we demonstrate them numerically in computer simulations under various standard statistical models.
\end{abstract}


\section{Introduction}
\label{sec:introduction}


In multiple test problems with composite null hypotheses, to account for type $\mathrm{I}$ errors, marginal tests are usually calibrated with respect to least favorable parameter configurations (LFCs). These are parameters values in (or on the boundary of) the corresponding null hypotheses under which the marginal tests are most likely to reject. Under certain assumptions, the resulting marginal LFC-based $p$-values are then uniformly distributed on $[0,1]$ ($\mathrm{Uni}[0,1]$-distributed) under LFCs, but stochastically larger than $\mathrm{Uni}[0,1]$ under non-LFCs in the null hypothesis. Under the alternative, LFC $p$-values usually tend to be stochastically smaller than $\mathrm{Uni}[0,1]$.

While the latter property is desirable in terms of protecting against type II errors, the deviation from uniformity under null hypotheses is problematic for some estimators of the proportion $\pi_{0}$ of true null hypotheses that use the empirical cumulative distribution function (ecdf) of all marginal $p$-values. We will denote the latter ecdf by $\hat{F}_{m}$ throughout the remainder, where $m$ is the number of all null hypotheses. One ecdf-based estimator for $\pi_{0}$ was introduced by \cite{schweder1982plots}, and it is given by
\begin{equation}
\label{eq:schwederspjotvoll}
\hat{\pi}_{0} \equiv \hat{\pi}_{0}(\lambda)=(1-\hat{F}_{m}(\lambda))/(1-\lambda),
\end{equation}  
where $\lambda\in(0,1)$ is a tuning parameter. The estimator $\hat{\pi}_{0}(\lambda)$ only works properly if the marginal $p$-values that correspond to the true null hypotheses are $\mathrm{Uni}[0,1]$-distributed. It is an unbiased estimator if all $p$-values that correspond to the false null hypotheses are smaller than $\lambda$ with probability one and all $p$-values that correspond to the true null hypotheses are $\mathrm{Uni}[0,1]$-distributed. Since (valid) $p$-values are stochastically not smaller than $\mathrm{Uni}[0,1]$ under the null, $\hat{\pi}_{0}(\lambda)$ is non-negatively biased. It is also known for a longer time (cf., e.\ g., the discussion by \cite{storey2004strong} after their Eq. (4)), that the variance of $\hat{\pi}_{0}(\lambda)$ increases with increasing $\lambda$ in most cases.

The aforementioned deviation from $\mathrm{Uni}[0,1]$ happens for instance in case of discrete models, which has been, among others, investigated by \cite{finner2007note}, \cite{habiger2011randomised}, \cite{dickhaus2012how} and \cite{habiger2015multiple}. In case of composite null hypotheses, the deviation of $p$-values from uniformity occurs, when marginal test statistics do not have a unique distribution under the null hypotheses and the marginal tests hence cannot be calibrated precisely with respect to their type I error probabilities. To provide more uniform $p$-values under composite null hypotheses \cite{dickhaus2013randomized} proposed randomized $p$-values that result from a data-dependent mixing of the LFC-based $p$-values and additional $\mathrm{Uni}[0,1]$-distributed random variables that are (stochastically) independent of the data. In certain models, these randomized $p$-values can be simplified to have a linear structure (cf. \cite{hoang2019randomized}).

While accurate estimations of $\pi_{0}$ are valuable in themselves, they can also improve the power of existing multiple test procedures. Namely, many of such procedures are (implicitly) calibrated to control the family-wise error rate (FWER) or the false discovery rate (FDR), respectively, for the case that every null hypothesis is true, that is, in case of $\pi_{0}=1$, which is often the worst case. If some null hypotheses are false, these procedures become over-conservative. Adjusting them according to a pre-estimate of $\pi_{0}$ can improve the overall power of these tests. \cite{benjamini2000adaptive} discuss these so-called adaptive procedures where the original procedure is the linear step-up test from \cite{benjamini1995controlling}. \cite{storey2003positive} proved that applying the linear step-up test by \cite{benjamini1995controlling} at an adjusted level controls the FDR if the $p$-values are independent. \cite{finner2009controlling} investigated the use of estimators of $\pi_{0}$ as plug-in estimators in single-step or step-down procedures and proved that the Bonferroni procedure at an adjusted level controls the FWER if the marginal $p$-values are independent. Further results and references on adaptive multiple tests (for FDR control) can be found in \cite{Heesen1,Heesen2}, and \cite{Heesen3}.

We focus on the case of composite null hypotheses and present a new way of randomizing LFC-based $p$-values. To this end, we utilize a set of stochastically independent and identically $\mathrm{Uni}[0,1]$-distributed random variables $U_{1},\ldots,U_{m}$, which are (stochastically) independent of the data $X$, as well as a set of constants $c_{1},\ldots,c_{m}$, where $c_j \in[0,1]$ for all $1 \leq j \leq m$. For an LFC-based $p$-value $p_{j}^{LFC}(X)$ we propose randomized $p$-values defined as
\begin{equation}
\label{eq:randomizedpvalues}
p_{j}^{rand}(X,U_{j},c_{j})=U_{j}\textbf{1}\{p_{j}^{LFC}(X)\geq c_{j}\}+p_{j}^{LFC}(X)c_{j}^{-1}\textbf{1}\{p_{j}^{LFC}(X)<c_{j}\},
\end{equation} 
$j=1\ldots,m$.

In many models this definition comprises the one of \cite{dickhaus2013randomized} for certain values of $c_{j}\in[0,1]$ (cf. \cite{hoang2019randomized}). It is clear that $c_{j}$ determines how close $p_{j}^{rand}$ is to either $U_{j}$ or $p_{j}^{LFC}$. The choices $c_{j}=0$ and $c_{j}=1$ lead to $p_{j}^{rand}=U_{j}$ or $p_{j}^{rand}=p_{j}^{LFC}$ (with probability one), respectively. Under certain conditions, it holds $U_{j}\leq_{\mathrm{st}}p_{j}^{rand}\leq_{\mathrm{st}}p_{j}^{LFC}$ under the $j$-th null hypothesis and $p_{j}^{LFC}\leq_{\mathrm{st}}p_{j}^{rand}\leq_{\mathrm{st}}U_{j}$ under the $j$-th alternative, where $\leq_{\mathrm{st}}$ denotes the stochastic order (see, e.\ g.,  Corollary \ref{cor:orderinexamples} below). While $\mathrm{Uni}[0,1]$-distributed $p$-values are desirable under null hypotheses, we want to keep them small under alternatives.
When using $p_{1}^{rand}(X,U_{1},c_{1}),\ldots,p_{m}^{rand}(X,U_{m},c_{m})$ in $\hat{\pi}_{0}$, we discuss how the choice of the constants $c_{1},\ldots,c_{m}$ affects the bias of $\hat{\pi}_{0}$. Under the restriction of identical $c_{j}$'s, we find that there exists a $c^{\star}\in[0,1]$ for which $\hat{\pi}_{0}$ has minimal bias when using $p_{1}^{rand}(X,U_{1},c^{\star}),\ldots,p_{m}^{rand}(X,U_{m},c^{\star})$.

The rest of the work is organized as follows. In Section$~\ref{sec:modelsetup}$ we provide the model framework. In Section$~\ref{sec:randomizedpvalues}$ we analyze properties of our proposed randomized $p$-values, and compare them to the LFC-based ones. Section$~\ref{sec:proportionoftruenullhypotheses}$ presents computer simulations to evaluate the performance of the proposed randomized $p$-values in estimating 
$\pi_{0}$. We conclude with a discussion in Section$~\ref{sec:discussion}$.

\section{Model Setup}
\label{sec:modelsetup}

We consider a statistical model $(\Omega, \mathcal{F}, (\mathbb{P}_{\vartheta})_{\vartheta\in\Theta})$, where $\vartheta$ denotes the parameter of the model and $\Theta$ the corresponding parameter space. In the context of multiple testing we define a derived parameter $\theta = \theta(\vartheta) =(\theta_{1}(\vartheta),\ldots,\theta_{m}(\vartheta))^\top$ with values in $\mathbb{R}^{m}$, $m\geq 2$. The $j$-th component $\theta_{j}(\vartheta)$ of this derived parameter is assumed to be the object of interest in the $j$-th null hypothesis $H_{j}$, $j=1,\ldots, m$, where the family of $m$ null hypotheses $H_{1},\ldots,H_{m}$ and the family of their corresponding alternatives $K_{1}, \ldots, K_{m}$ consist of non-empty Borel sets of $\mathbb{R}$. For each $j=1, \ldots, m$ we test $\theta_{j}(\vartheta)\in H_{j}$ against $\theta_{j}(\vartheta)\in K_{j}=\mathbb{R}\setminus H_{j}$.

We assume that for each $j=1,\ldots,m$ a test statistic $T_{j}:\Omega\to\mathbb{R}$ and a rejection region 
$\Gamma_{j}(\alpha) \subset \mathbb{R}$ are given, where $\alpha\in(0,1)$ denotes a fixed, local significance level. We denote by $x\in\Omega$ the realization of $X$. The test statistics $\{T_{j}(X)\}_{1 \leq j \leq m}$ are assumed to have absolutely continuous distributions with respect to the Lebesgue measure under any $\vartheta\in\Theta$. The marginal tests $\varphi_{j}$ for testing $H_{j}$ versus $K_{j}$ are given by 
$\varphi_{j}(X)=\textbf{1}\{T_{j}(X)\in\Gamma_{j}(\alpha)\}$, where $\varphi_j(x) = 1$ means rejection of $H_j$ in favor of $K_j$ and $\varphi_j(x) = 0$ means that $H_j$ is retained, for observed data $x$ and $1\leq j \leq m$. Note, that we do not make any (general) assumptions about the dependency structure among the different test statistics at this point.

Furthermore, we make the following additional general assumptions:
\begin{description}
\item[$(A1)$] Nested rejection regions: For every $j=1,\ldots,m$ and $\alpha'<\alpha$, it holds that $\Gamma_{j}(\alpha')\subseteq\Gamma_{j}(\alpha)$.
\item[$(A2)$] For every $j=1,\ldots,m$, it holds $\underset{\vartheta:\theta_{j}(\vartheta)\in H_{j}}{\mathrm{sup}}\mathbb{P}_{\vartheta}(T_{j}(X)\in\Gamma_{j}(\alpha))=\alpha$.
\item[$(A3)$] The set of LFCs for $\varphi_{j}$, i.\ e., the set of parameter values that yield the supremum in $(A2)$, does not depend on $\alpha$.
\end{description}

Under assumption $(A1)$, rejections at significance levels $\alpha'$ always imply rejections at larger significance levels $\alpha>\alpha'$. Assumption $(A2)$ means that under any LFC for $\varphi_{j}$ the rejection probability is exactly $\alpha$.

LFC-based $p$-values for the marginal tests $\{\varphi_{j}\}_{1 \leq j \leq m}$ are formally defined as
\begin{equation*}
p_{j}^{LFC}(X)=\underset{\{\tilde{\alpha}\in(0,1):T_{j}(x)\in\Gamma_{j}(\tilde{\alpha})\}}{\mathrm{inf}}\;\underset{\{\vartheta:\theta_{j}(\vartheta)\in H_{j}\}}{\mathrm{sup}}\mathbb{P}_{\vartheta}(T_{j}(X)\in\Gamma_{j}(\tilde{\alpha})).
\end{equation*}
Under assumptions $(A1)$ -- $(A3)$, we obtain that 
\begin{equation}\label{plfc}
p_{j}^{LFC}(X)={\mathrm{inf}\{\tilde{\alpha}\in(0,1):T_{j}(X)\in\Gamma_{j}(\tilde{\alpha})\}}, \; j=1,\ldots,m.
\end{equation} 
With assumption $(A2)$, any such LFC-based $p$-value $p_{j}^{LFC}(X)$ is uniformly distributed on $[0,1]$ under any LFC for $\varphi_{j}$; cf. Lemma 3.3.1 of 
\cite{lehmann2005testing}. Let $F_{\vartheta}$ be the cumulative distribution function (cdf) of $T_{j}(X)$ under $\vartheta \in \Theta$. If the rejection region $\Gamma_{j}(\alpha)$ is given by $(F_{\vartheta_{0}}^{-1}(1-\alpha),\infty)$, where $\vartheta_{0}$ is an LFC for $\varphi_{j}$, then the definition in 
\eqref{plfc} simplifies to $p_{j}^{LFC}(X)=1-F_{\vartheta_{0}}(T_{j}(X))$. Rejection regions of that type are typical if test statistics tend to larger values under alternatives, which is often the case.

As examples, we give two models that fulfill the general assumptions $(A1)$ -- $(A3)$.

\begin{example}[Multiple $Z$-tests model]\label{ex:multipleztests}
We consider $X=(X_{i,j}:i=1,\ldots,n_{j},\,j=1,\ldots,m)$, where $(n_{j})_{j=1,\ldots,m}$ are fixed sample sizes. For all $j$ the random variables $X_{1,j},\ldots,X_{n_{j},j}$ are assumed to be stochastically independent and identically normally distributed as $N(\theta_{j}(\vartheta),1)$, where $\vartheta=(\vartheta_{1},\ldots,\vartheta_{m})^\top\in\Theta=\mathbb{R}^{m}$ is the (main) parameter of the model and $\theta(\vartheta)$, given by $\theta_{j}(\vartheta)=\vartheta_{j}$ for $1 \leq j \leq m$, is the derived parameter. For each $1 \leq j \leq m$, we are interested in the null hypothesis $H_{j}:\vartheta_{j}\leq 0$ against its alternative $K_{j}:\vartheta_{j}>0$, and consider the test statistic $T_{j}(X)=n_{j}^{-1}\sum_{i=1}^{n_{j}}X_{i,j}\sim N(\vartheta_{j},n_{j}^{-1})$. Furthermore, we let $\Gamma_{j}(\alpha)=(\Phi_{(0,n_{j}^{-1})}^{-1}(1-\alpha),\infty)$, leading to the LFC-based $p$-value $p_{j}^{LFC}(X)=1-\Phi_{(0,n_{j}^{-1})}(T_{j}(X))$, where $\Phi_{(\mu, \sigma^2)}$ denotes the cdf of the normal distribution on $\mathbb{R}$ with parameters $\mu$ and $\sigma^2$.  For each $j=1,\ldots,m$, the set of LFCs for $\varphi_{j}$ is $\{\vartheta\in\Theta:\vartheta_{j}=0\}$, independently of $\alpha$. As mentioned before, we do not specify the dependency structure of $T_{j_1}(X)$ and $T_{j_2}(X)$ for $1 \leq j_1 \neq j_2 \leq m$. The latter dependency structure may be regarded as a further (nuisance) parameter of the model.
\end{example}

\begin{example}[Two-sample means comparison model]\label{ex:twosamples}
Let $j=1,\ldots,m$ be fixed. For given sample sizes $n_{1,j}$ and $n_{2,j}$, let $X_{1,j},\ldots,X_{n_{1,j},j}$ and $Y_{1,j},\ldots,Y_{n_{2,j},j}$ be jointly stochastically independent, observable random variables. Assume that $X_{1,j},\ldots,$ $X_{n_{1,j},j}$ are identically distributed with $X_{1,j} \sim N(\theta_{1,j}(\vartheta),\sigma^2_{j})$, and that $Y_{1,j},\ldots,$ $Y_{n_{2,j},j}$ are identically distributed with $Y_{1,j} \sim N(\theta_{2,j}(\vartheta),\sigma^2_{j})$, where $\sigma^2_{j}>0$ is unknown. Similarly as in Example \ref{ex:multipleztests}, the parameter vector $\vartheta$ consists of all unknown means and all unknown variances of the model. For each $1 \leq j \leq m$, we compare the means of the two samples. To this end, we let $\theta_{j}(\vartheta) = \theta_{1,j}(\vartheta) - \theta_{2,j}(\vartheta)$ and assume that $H_{j}:\theta_{j}(\vartheta) \leq 0$ versus $K_{j}: \theta_{j}(\vartheta) > 0$ is the marginal test problem of interest. Let $\bar{X}_{j}=n_{1,j}^{-1}\sum_{i=1}^{n_{1,j}}X_{i,j}$, $\bar{Y}_{j}=n_{2,j}^{-1}\sum_{i=1}^{n_{2,j}}Y_{i,j}$, and 
\begin{equation*}
S_{j}(X)=\frac{1}{n_{1,j}+n_{2,j}-2}\Big[\sum_{i=1}^{n_{1,j}}(X_{i,j}-\bar{X}_{j})^{2}+\sum_{i=1}^{n_{2,j}}(Y_{i,j}-\bar{Y}_{j})^{2}\Big].
\end{equation*}
Under an LFC for $\varphi_{j}$, that is, any $\vartheta\in\Theta$ with $\theta_{j}(\vartheta)=0$, the test statistic
\begin{equation*}
T_{j}(X)=\sqrt{\frac{n_{1,j}n_{2,j}}{n_{1,j}+n_{2,j}}}(\bar{X}_{j}-\bar{Y}_{j})/S_{j}
\end{equation*} 
follows Student's $t$-distribution with $n_{1,j}+n_{2,j}-2$ degrees of freedom, denoted by $t_{n_{1,j}+n_{2,j}-2}$. The corresponding rejection region is $\Gamma_{j}(\alpha)=(F_{t_{n_{1,j}+n_{2,j}-2}}^{-1}(1-\alpha),\infty)$ and the LFC-based $p$-value is given by $p_{j}^{LFC}(X)=1-F_{t_{n_{1,j}+n_{2,j}-2}}(T_{j}(X))$, where $F_{t_{n_{1,j}+n_{2,j}-2}}$ denotes the cdf of $t_{n_{1,j}+n_{2,j}-2}$. Again, the aforementioned set of LFCs for $\varphi_j$ does not depend on $\alpha$, for each $1 \leq j \leq m$. For the dependency structure among different coordinates $j_1 \neq j_2$, we argue as in Example \ref{ex:multipleztests}.
\end{example}

\section{The randomized $p$-values}\label{sec:randomizedpvalues}
 
\subsection{General properties}\label{sec:generalproperties}

\begin{definition}\label{def:randomizedpvalues}
Let a model as in Section$~\ref{sec:modelsetup}$ and a set of random variables $U_{1},\ldots,U_{m}$, that are defined on the same probability space as $X$, jointly stochastically independent, identically $\mathrm{Uni}[0,1]$-distributed (under any $\vartheta \in \Theta$), and stochastically independent of the data $X$, be given. For each $j=1,\ldots,m$ and given constants $c_{1},\ldots,c_{m}$ with $c_j \in [0,1]$ for all $1 \leq j \leq m$, we define our randomized $p$-values as in Equation \eqref{eq:randomizedpvalues}, where $p_{j}^{rand}(X,U_{j}, 0)=U_{j}$ by convention.
\end{definition}

For a more general definition of these $p$-values, we refer to the appendix. Before we discuss the properties of these randomized $p$-values and compare them to LFC-based ones, we give a few remarks.

\begin{remark} $ $
\begin{description}
\item[$(a.)$] If $p_{j}^{LFC}(X)$ is stochastically large, then  it is likely that $p_{j}^{rand}(X,U_{j},c_{j})=U_{j}$ holds. This means that under the null hypothesis $H_{j}$, the distribution of $p_{j}^{rand}$ will typically be close to a $\mathrm{Uni}[0,1]$-distribution. On the other hand, if $K_{j}$ is true and $p_{j}^{LFC}(X)$ is stochastically small, the randomized $p$-value $p_{j}^{rand}(X,U_{j},c_{j})$ is more likely to be equal to $p_{j}^{LFC}(X)/c_{j}\geq p_{j}^{LFC}(X)$ than it is to be equal to $U_{j}$.
\item[$(b.)$] Under an LFC $\vartheta_{0}$ for $\varphi_{j}$ the randomized $p$-value $p_{j}^{rand}(X,U_{j},c_{j})$ is uniformly distributed on $[0,1]$ for any $1 \leq j \leq m$. Namely, it holds that
\begin{align*}
\mathbb{P}_{\vartheta_{0}}(p_{j}^{rand}(X,U_{j},c_{j})\leq t)&=\mathbb{P}_{\vartheta_{0}}(U_{j}\leq t)\,\mathbb{P}_{\vartheta_{0}}(p_{j}^{LFC}(X)\geq c_{j})+\mathbb{P}_{\vartheta_{0}}(p_{j}^{LFC}(X)<t c_{j})\\
&=t(1-c_{j})+t\,c_{j}=t,
\end{align*}
where we have used that $p_{j}^{LFC}(X)$ is $\mathrm{Uni}[0,1]$-distributed under any LFC $\vartheta_{0}$ for $\varphi_{j}$, due to assumptions $(A1)$ -- $(A2)$, and that $U_j$ is always $\mathrm{Uni}[0,1]$-distributed, no matter the value of $\vartheta$.
\end{description}
\end{remark}


As mentioned in Section$~\ref{sec:introduction}$, the use of valid $p$-values in the Schweder-Spj\o tvoll estimator ensures that the latter has a non-negative bias; cf. Lemma 1 of \cite{dickhaus2012how}. Therefore it is of interest to give some conditions for the validity of our randomized $p$-values.

\begin{theorem}\label{thm:randpvalidity}
Let a model as in Section$~\ref{sec:modelsetup}$ be given and $j\in\{1,\ldots,m\}$ be fixed. 
Then, $p_{j}^{rand}(X,U_{j},c_{j})$ is a valid $p$-value for a given $c_{j}\in[0,1]$ if and only if the following condition (1.) is fulfilled. Furthermore, either of the following conditions (2.) and (3.) is a sufficient condition for the validity of $p_{j}^{rand}(X,U_{j},c_{j})$ for any $c_j \in [0, 1]$.

\begin{description}
\item[$(1.)$] For every $\vartheta\in\Theta$ with $\theta_{j}(\vartheta)\in H_{j}$, it holds
\begin{equation*}
\mathbb{P}_{\vartheta}(p_{j}^{LFC}(X)\leq t\,c_{j})\leq t\,\mathbb{P}_{\vartheta}(p_{j}^{LFC}(X)\leq c_{j})
\end{equation*}
for all $t\in[0,1]$.
\item[$(2.)$] For every $\vartheta\in\Theta$ with $\theta_{j}(\vartheta)\in H_{j}$, $\mathbb{P}_{\vartheta}(p_{j}^{LFC}(X)\leq t)/t$ is non-decreasing in $t$.
\item[$(3.)$] The cdf of $p_{j}^{LFC}(X)$ is convex under any parameter $\vartheta\in\Theta$ with $\theta_{j}(\vartheta)\in H_{j}$.
\end{description}

If the LFC-based $p$-value is given by $p_{j}^{LFC}(X)=1-F_{\vartheta_{0}}(T_{j}(X))$, where $\vartheta_{0}\in\Theta$ is an LFC for $\varphi_{j}$, then the following condition $(4.)$ is equivalent to condition 
$(2.)$, while condition $(5.)$ is equivalent to condition $(3.)$.

\begin{description}
\item[$(4.)$] For every $\vartheta\in\Theta$ with $\theta_{j}(\vartheta)\in H_{j}$, it holds $T_{j}(X)^{(\vartheta)}\leq_{\mathrm{hr}}T_{j}(X)^{(\vartheta_{0})}$.
\item[$(5.)$] For every $\vartheta\in\Theta$ with $\theta_{j}(\vartheta)\in H_{j}$, it holds $T_{j}(X)^{(\vartheta)}\leq_{\mathrm{lr}}T_{j}(X)^{(\vartheta_{0})}$.
\end{description}
With $\leq_{\mathrm{hr}}$ and $\leq_{\mathrm{lr}}$ we mean the hazard rate order and the likelihood ratio order, respectively. The notation $T_{j}(X)^{(\vartheta)}$ refers to the distribution of $T_{j}(X)$ under $\vartheta\in\Theta$. The relationship $T_{j}(X)^{(\vartheta)}\leq_{\mathrm{hr}}T_{j}(X)^{(\vartheta_{0})}$ is equivalent to $(1-F_{\vartheta_{0}}(t))/(1-F_{\vartheta}(t))$ being non-decreasing in $t$, and $T_{j}(X)^{(\vartheta)}\leq_{\mathrm{lr}}T_{j}(X)^{(\vartheta_{0})}$ is equivalent to $f_{\vartheta_{0}}(t)/f_{\vartheta}(t)$ being non-decreasing in $t$, where $f_{\vartheta}$ denotes the Lebesgue density of $T_{j}(X)$ under $\vartheta\in\Theta$.
\end{theorem}
The proof of Theorem \ref{thm:randpvalidity} is given in the appendix.

\begin{corollary}\label{lm:thm1lemma1}
Under the models from Examples$~\ref{ex:multipleztests}$ and $\ref{ex:twosamples}$, the randomized $p$-values $(p_{j}^{rand}(X,U_{j},c_{j}))_{1 \leq j \leq m}$ are valid for any $(c_{1},\ldots,c_{m})^\top \in[0,1]^m$.
\end{corollary}

\begin{proof}
The multiple $Z$-tests model from Example$~\ref{ex:multipleztests}$ fulfills the general assumptions $(A1)$ -- $(A3)$ from Section$~\ref{sec:modelsetup}$. Let $j \in \{1, \hdots, m\}$ be arbitrarily chosen. For a parameter value $\vartheta\in\Theta$ with $\theta_{j}(\vartheta)=\vartheta_{j}\in H_{j}$, i.\ e., $\vartheta_{j}\leq 0$, it is easy to show that $f_{0}(t)/f_{\vartheta_{j}}(t)$ is non-decreasing in $t$, where $f_{z}$ denotes the Lebesgue density of the $N(z,n_{j}^{-1})$-distribution. Following Theorem$~\ref{thm:randpvalidity}$, $p_{j}^{rand}(X,U_{j},c_{j})$ is valid for any constant $c_{j} \in [0,1]$. The choice of $c_{j}=1/2$ for all $1 \leq j \leq m$ results in the randomized $p$-values from \cite{dickhaus2013randomized} for this model.

The two-sample means comparison model from Example$~\ref{ex:twosamples}$ fulfills the general assumptions $(A1)$ -- $(A3)$, too. Again, let $j \in \{1, \hdots, m\}$ be arbitrarily chosen. Under any parameter value $\vartheta\in\Theta$ it holds that $T_{j}(X)\sim t_{\tau_{j},n_{1,j}+n_{2,j}-2}$, where 
$\tau_{j}=\sqrt{\frac{n_{1,j}n_{2,j}}{n_{1,j}+n_{2,j}}}\theta_{j}(\vartheta)/\sigma_{j}$, and $t_{\tau,\nu}$ denotes the non-central $t$-distribution with non-centrality parameter $\tau$ and $\nu$ degrees of freedom. The family\\ $(t_{\tau,n_{1,j}+n_{2,j}-2})_{\tau\in\mathbb{R}}$ of distributions possesses the monotone likelihood ratio (MLR) property, i.\ e., it holds $t_{\tau_{1},n_{1,j}+n_{2,j}-2}\leq_{\mathrm{lr}}t_{\tau_{2},n_{1,j}+n_{2,j}-2}$ if and only if $\tau_{1}\leq\tau_{2}$; cf. \cite{karlin1956decision} and \cite{karlin1956distributions}. For a parameter value $\vartheta\in\Theta$ with $\theta_{j}(\vartheta)\in H_{j}$, i.\ e., $\theta_{1,j}(\vartheta)\leq\theta_{2,j}(\vartheta)$, it holds that $\tau_{j}=\sqrt{\frac{n_{1,j}n_{2,j}}{n_{1,j}+n_{2,j}}}\theta_{j}(\vartheta)/\sigma_{j}\leq 0$ and therefore $T_{j}(X)^{(\vartheta)}\leq_{\mathrm{lr}}T_{j}(X)^{(\vartheta_{0})}$, where $\vartheta_{0}$ is an LFC for $\varphi_{j}$, i.\ e., $\theta_{1,j}(\vartheta_{0})=\theta_{2,j}(\vartheta_{0})$. According to Theorem$~\ref{thm:randpvalidity}$, $p_{j}^{rand}(X,U_{j},c_{j})$ is valid for any choice of the constant $c_{j}\in[0,1]$ in this model.
\end{proof}

\subsection{A comparison between the LFC-based and the randomized $p$-values}
\label{sec:comparisonpvalues}
For any $1 \leq j \leq m$, we want to compare the cdfs of $p_{j}^{LFC}(X)$ and $p_{j}^{rand}(X,U_{j},c_{j})$.
Due to the discussion below \eqref{eq:randomizedpvalues}, this comparison is trivial for $c_j = 0$ and for $c_j = 1$, respectively. Therefore, let us assume here that $c_j$ is bounded away from zero and from one. For example, one may for the moment assume that $c_j = 0.5$ is chosen, for concreteness.

We first note that 
\begin{align}
\mathbb{P}_{\vartheta}(p_{j}^{LFC}(X)\leq t)=&\mathbb{P}_{\vartheta}(p_{j}^{LFC}(X)\leq t\mid p_{j}^{LFC}(X)>c_{j})\,\mathbb{P}_{\vartheta}(p_{j}^{LFC}(X)>c_{j})\nonumber\\
&+\mathbb{P}_{\vartheta}(p_{j}^{LFC}(X)\leq t,\,p_{j}^{LFC}(X)\leq c_{j}),\label{eq:comparison1}\\
\mathbb{P}_{\vartheta}(p_{j}^{rand}(X,U_{j},c_{j})\leq t)=&\mathbb{P}_{\vartheta}(U_{j}\leq t)\,\mathbb{P}_{\vartheta}(p_{j}^{LFC}(X)>c_{j})+\mathbb{P}_{\vartheta}(p_{j}^{LFC}(X)\leq t c_{j}).\label{eq:comparison2}
\end{align}
Now, if the value of the derived parameter $\theta_{j}(\vartheta)$ is so "deep inside" $H_{j}$ that $\mathbb{P}_{\vartheta}(p_{j}^{LFC}(X)> c_j)$ is large, then the first summands in $\eqref{eq:comparison1}$ and $\eqref{eq:comparison2}$ dominate the second ones, and we see that
\begin{equation*}
\mathbb{P}_{\vartheta}(p_{j}^{LFC}(X)\leq t\mid p_{j}^{LFC}(X)>c_{j})\leq \mathbb{P}_{\vartheta}(p_{j}^{LFC}(X)\leq t)\leq t=\mathbb{P}_{\vartheta}(U_{j}\leq t).
\end{equation*}
Thus, provided that $p_{j}^{rand}(X,U_{j},c_{j})$ is a valid $p$-value, its distribution under $H_j$ will typically be closer to $\mathrm{Uni}[0,1]$ than that of $p_{j}^{LFC}(X)$.

However, if $\vartheta$ is such that $K_{j}$ is true instead and that  $\mathbb{P}_{\vartheta}(p_{j}^{LFC}(X)\leq c_{j})$ is large, it holds that
\begin{align*}
\mathbb{P}_{\vartheta}(p_{j}^{LFC}(X)\leq t,\,p_{j}^{LFC}(X)\leq c_{j})&=\mathbb{P}_{\vartheta}(p_{j}^{LFC}(X)\leq\mathrm{min}(t,c_{j}))\\
&\geq\mathbb{P}_{\vartheta}(p_{j}^{LFC}(X)\leq t c_{j}).
\end{align*}
Thus, under $K_j$ the cdf of $p_{j}^{LFC}(X)$ will typically be pointwise larger than the cdf of $p_{j}^{rand}(X,U_{j},c_{j})$.

The former heuristic argumentation cannot be made mathematically rigorous in general. However, if condition $(3.)$ in Theorem$~\ref{thm:randpvalidity}$ is fulfilled, $p_{j}^{rand}$ does indeed always lie between $U_j$ and $p_{j}^{LFC}$ under the null hypothesis $H_{j}$, in the sense of the stochastic order. The same holds under the alternative $K_{j}$, if a condition similar to $(3.)$ is fulfilled in the case of 
$\theta_j(\vartheta) \in K_j$.

\begin{theorem}\label{thm:stochasticorder}
Let a model as in Section$~\ref{sec:modelsetup}$ be given and $j\in\{1,\ldots,m\}$ be fixed. 

If the cdf of $p_{j}^{LFC}(X)$ is convex under a fixed $\vartheta\in\Theta$, then 
\begin{equation*}
p_{j}^{rand}(X,U_{j},c_{j})^{(\vartheta)}\leq_{\mathrm{st}}p_{j}^{rand}(X,U_{j},\tilde{c}_{j})^{(\vartheta)}
\end{equation*}
for any $0\leq c_{j}\leq\tilde{c}_{j}\leq 1$.

If the cdf of $p_{j}^{LFC}(X)$ is concave under a fixed $\vartheta\in\Theta$, then it holds that
\begin{equation*}
p_{j}^{rand}(X,U_{j},\tilde{c}_{j})^{(\vartheta)}\leq_{\mathrm{st}}p_{j}^{rand}(X,U_{j},c_{j})^{(\vartheta)}
\end{equation*}
for any $0\leq c_{j}\leq\tilde{c}_{j}\leq 1$.
\end{theorem}

We give the proof of Theorem \ref{thm:stochasticorder} in the appendix.

\begin{remark} \label{rmk:thm2rmk1}
Let $j\in\{1,\ldots,m\}$ be fixed.
\begin{description}
\item[$1.$] If the $j$-th LFC-based $p$-value is given by $p_{j}^{LFC}(X)=1-F_{\vartheta_{0}}(T_{j}(X))$, where $\vartheta_{0}$ is an LFC for $\varphi_{j}$, then $p_{j}^{LFC}(X)$ has a convex cdf under $\vartheta\in\Theta$ if and only if $T_{j}(X)^{(\vartheta)}\leq_{\mathrm{lr}}T_{j}(X)^{(\vartheta_{0})}$, and a concave cdf under $\vartheta\in\Theta$ if and only if $T_{j}(X)^{(\vartheta_{0})}\leq_{\mathrm{lr}}T_{j}(X)^{(\vartheta)}$ (cf. the proof of Theorem$~\ref{thm:randpvalidity}$ in the appendix).
\item[$2.$] If condition $(3.)$ from Theorem$~\ref{thm:randpvalidity}$ is fulfilled, then Theorem$~\ref{thm:stochasticorder}$ implies
\begin{equation*}
U_j \leq_{\mathrm{st}}p_{j}^{rand}(X,U_{j},c_{j})^{(\vartheta)}\leq_{\mathrm{st}}p_{j}^{LFC}(X)^{(\vartheta)}
\end{equation*}  
for all $\vartheta\in\Theta$ with $\theta_{j}(\vartheta)\in H_{j}$ and any $c_{j}\in[0,1]$. This also implies the validity of $p_{j}^{rand}(X,U_{j},c_{j})$, as it was claimed in Theorem$~\ref{thm:randpvalidity}$.
\item[$3.$] The cdf of $p_{j}^{LFC}(X)$ can never be concave under $H_{j}$.
\end{description}
\end{remark}

\begin{corollary}\label{cor:orderinexamples}
For the multiple $Z$-tests model and the two-sample means comparison model from Examples$~\ref{ex:multipleztests}$ and $\ref{ex:twosamples}$, respectively, it holds for any $1 \leq j \leq m$  and any $c_{j}\in[0,1]$, that
\begin{equation*}
U_j\leq_{\mathrm{st}}p_{j}^{rand}(X,U_{j},c_{j})^{(\vartheta)}\leq_{\mathrm{st}}p_{j}^{LFC}(X)^{(\vartheta)}
\end{equation*}  
under any $\vartheta \in \Theta$ with $\theta_{j}(\vartheta)\in H_{j}$, as well as 
\begin{equation*}
p_{j}^{LFC}(X)^{(\vartheta)}\leq_{\mathrm{st}}p_{j}^{rand}(X,U_{j},c_{j})^{(\vartheta)}\leq_{\mathrm{st}}U_j
\end{equation*}  
under any $\vartheta$ with $\theta_{j}(\vartheta)\in K_{j}$. 
\end{corollary}
We conclude this section by illustrating the assertions of Theorem \ref{thm:stochasticorder} and Corollary \ref{cor:orderinexamples} under the multiple $Z$-tests model. In Figures$~\ref{fig:cdfsnull}$ and $\ref{fig:cdfsalternative}$ we compare the cdfs of $p_{j}^{rand}(X,U_{j},c)$ for an arbitrary $j \in \{1, \hdots, m\}$ for $c=0,0.25,0.5,0.75$, and $1$ under $\vartheta\in\Theta$, where we set $\theta_{j}(\vartheta)=-1/\sqrt{n_{j}}$ or $\theta_{j}(\vartheta)=1/\sqrt{n_{j}}$ for $n_{j}=50$, respectively. It is apparent that the cdfs move from that of the $\mathrm{Uni}[0,1]$-distribution to the one of $p_{j}^{LFC}(X)$ with increasing $c$.

\section{Estimation of the proportion of true null hypotheses}
\label{sec:proportionoftruenullhypotheses}

\subsection{Some theoretical results}\label{sec:pi0theory}
We consider the usage of $\{p_{j}^{rand}(X,U_{j},c_{j})\}_{1 \leq j \leq m}$ in the Schweder-Spj\o tvoll estimator $\hat{\pi}_{0} \equiv \hat{\pi}_{0}(\lambda)$ defined in \eqref{eq:schwederspjotvoll}. It can easily be seen from the representation on the right-hand side of \eqref{eq:schwederspjotvoll}, that the bias $\mathbb{E}_{\vartheta}[\hat{\pi}_{0}(\lambda)]-\pi_{0}$ of $\hat{\pi}_{0}(\lambda)$ decreases if $\mathbb{E}_{\vartheta}[\hat{F}_{m}(\lambda)]$ increases, under any $\vartheta \in \Theta$. Thus, in terms of bias reduction of $\hat{\pi}_{0}(\lambda)$ (for a fixed, given value of $\lambda$) stochastically small (randomized) $p$-values (with pointwise large cdfs) are most suitable. In order to avoid a negative bias of $\hat{\pi}_{0}(\lambda)$, we furthermore have to ensure validity of the $p$-values utilized in $\hat{\pi}_{0}(\lambda)$. Hence, if the cdfs of the LFC-based $p$-values are convex under null hypotheses and concave under alternatives, the optimal ("oracle") value of $c_{j}$ is zero whenever $H_{j}$ is true and one whenever $K_{j}$ is true; cf. Theorem \ref{thm:stochasticorder}. This is also in line with Remark 6 of \cite{dickhaus2012how}, who showed that $\hat{\pi}_{0}(\lambda)$ is unbiased if the $p$-values utilized in $\hat{\pi}_{0}(\lambda)$ are $\mathrm{Uni}[0,1]$-distributed under true null hypotheses and almost surely smaller than $\lambda$ under false null hypotheses. Under the restriction of identical $c_{j}$'s, i.\ e., $c_1 = c_2 = \hdots = c_m \equiv c$, one may expect that an optimal ("oracle") value of $c$ (leading to a small, but non-negative bias of $\hat{\pi}_{0}(\lambda)$) should be close to $1 - \pi_0$. The latter restriction will be made throughout the remainder for computational convenience and feasibility. 

\begin{definition}
The Schweder-Spj\o tvoll estimator $\hat{\pi}_{0}(\lambda)$, if used with $p_{1}^{rand}(X,U_{1},c)$, 
$\ldots,p_{m}^{rand}(X,U_{m},c)$, will be denoted by $\hat{\pi}_{0}(\lambda,c)$ throughout the remainder. Notice, that in the estimators $\hat{\pi}_{0}(\lambda,0)$ and $\hat{\pi}_{0}(\lambda,1)$, respectively, we use $U_{1},\ldots,U_{m}$ (as the marginal $p$-values) and $p_{1}^{LFC}(X),\ldots,p_{m}^{LFC}(X)$, respectively. Furthermore, we consider the function $h_{\vartheta}:[0,1)\times[0,1]\rightarrow\mathbb{R}$, given by $h_{\vartheta}(\lambda,c)=\mathbb{E}_{\vartheta}[\hat{\pi}_{0}(\lambda,c)]$, where 
$\vartheta\in\Theta$ is the underlying parameter value.
\end{definition}

\begin{lemma}
For every $\lambda\in[0,1)$ and under any $\vartheta \in \Theta$, $h_{\vartheta}(\lambda,0)=1$. If the cdfs of the 
$p_{1}^{LFC}(X),\ldots,p_{m}^{LFC}(X)$ are continuous under $\vartheta$, then there exists a minimizing argument $c^{\star}\in[0,1]$ of $h_{\vartheta}(\lambda,\cdot)$. 
\end{lemma}

\begin{proof}
In the case of $c=0$, $p_{j}^{rand}(X,U_{j},0)=U_{j}$ for each $j \in \{1,\hdots,m\}$, and $\mathbb{E}_{\vartheta}[\hat{\pi}_{0}(\lambda,0)]=(1-\lambda)/(1-\lambda)=1$, proving the first assertion. 

In order to show the second assertion, we note that under any $\vartheta \in \Theta$
\begin{equation}\label{eq:EFm}
\mathbb{E}_{\vartheta}[\hat{F}_{m}(\lambda)]=\sum_{j=1}^{m}\left[\lambda\,\mathbb{P}_{\vartheta}(p_{j}^{LFC}(X)\geq c)+\mathbb{P}_{\vartheta}(p_{j}^{LFC}(X)\leq c \lambda)\right].
\end{equation}
The right-hand side of \eqref{eq:EFm} is continuous in $c$ if the cdfs of the $p$-values $p_{1}^{LFC}(X)$,
$\ldots,p_{m}^{LFC}(X)$ are continuous under $\vartheta$. Since $[0,1]$ is a compact set, the function $h_{\vartheta}(\lambda,\cdot)$ attains a minimum on $[0,1]$, by the extreme value theorem.
\end{proof}

For an illustration, let us consider the multiple $Z$-tests model from Example$~\ref{ex:multipleztests}$, where we set the total number of null hypotheses to $m = 1{,}000$ and the sample sizes to $n_{j}=50$ for all $j=1,\ldots,m$. As mentioned before, the choice of $c=1/2$ leads to the randomized $p$-values as defined in \cite{dickhaus2013randomized} for this model. Figures 
\ref{fig:Epi0example1} and \ref{fig:Epi0example3} display the graphs of the function $c\mapsto h_{\vartheta}(1/2, c)$ for two different parameter values $\vartheta\in\Theta$ under this model. In both cases, 
$\pi_{0}=0.7$ (meaning that $700$ null hypotheses are true and $300$ are false) and $\theta_{j}(\vartheta)=2.5/\sqrt{50}$ whenever $H_j$ is false.

In Figure \ref{fig:Epi0example1}, $\theta_{j}(\vartheta)=-1/\sqrt{50}$ whenever $H_{j}$ is true. The minimum of $c\mapsto h_{\vartheta}(1/2, c)$ is attained at $c^{\star}=0.3276$ and yields $\mathbb{E}_{\vartheta}[\hat{\pi}_{0}(1/2,c^{\star})]=0.7508$. It is apparent that $h_{\vartheta}(1/2,c)$ is largest for $c=1$, that is, when utilizing the LFC-based $p$-values $\{p_{j}^{LFC}(X)\}_{1 \leq j \leq m}$. Furthermore, Figure \ref{fig:Epi0example1} graphically confirms, that $c^{\star}$ (indicated by the dashed vertical line) is close to $1 - \pi_0 = 0.3$ (indicated by the solid vertical line), as mentioned previously. Finally, we see that the optimal bias of $\hat{\pi}_{0}(1/2)$ when using the same $c_j \equiv c$ for all $1 \leq j \leq m$ is larger than zero (compare the dashed and the dotted horizontal lines).

In Figure \ref{fig:Epi0example3}, $\theta_{j}(\theta)=0$ whenever $H_{j}$ is true. In this case, the estimator $\hat{\pi}_{0}(1/2,1)$ has the lowest bias among all estimators $\{\hat{\pi}_{0}(1/2,c): c\in[0,1]\}$, meaning that $c^{\star}=1$. This is because for every $j$ with $\theta_{j}(\vartheta)\in H_{j}$, $\vartheta$ is an LFC for $\varphi_{j}$ and thus $p_{j}^{LFC}(X)$ is $\mathrm{Uni}[0,1]$-distributed under $\vartheta$. In such cases, $p_{j}^{rand}(X,U_{j},c)$ is $\mathrm{Uni}[0,1]$-distributed for any $c$ under $H_j$, while $p_{j}^{LFC}(X)^{(\vartheta)}\leq_{\mathrm{st}}p_{j}^{rand}(X,U_{j},c)^{(\vartheta)}$ if $K_{j}$ is true, due to Theorem \ref{thm:stochasticorder}.

From a decision-theoretic perspective, the bias alone is not enough to judge the estimation quality of 
$\hat{\pi}_{0}$. A more commonly used criterion for the quality of an estimator is its mean squared error (MSE), which equals the squared bias plus the variance of the estimator under consideration. 
Therefore, we now additionally discuss the variance of $\hat{\pi}_{0}$ when employing our proposed randomized $p$-values. As is apparent from the right-hand side of \eqref{eq:EFm}, $\mathbb{E}_{\vartheta}[\hat{F}_{m}(\lambda)]$ only depends on the marginal distributions of $p_{1}^{LFC}(X),\ldots,p_{m}^{LFC}(X)$, but not on their dependency structure (i.\ e., their copula). Consequently, also $\mathbb{E}_{\vartheta}[\hat{\pi}_{0}]$ and $h_{\vartheta}$ do not depend on that copula. 
However, the variance of $\hat{\pi}_{0}$ does depend on the dependency structure among the utilized $p$-values. In particular, in prior work (see \cite{neumann2017estimating}) it has been shown that a high degree of positive dependency among the utilized $p$-values entails a large variance of $\hat{\pi}_{0}$. 

Figures \ref{fig:varianceidandd} and \ref{fig:MSEidandd} illustrate the effect of the copula of the $p$-values 
utilized in $\hat{\pi}_{0}$ on its variance and its MSE, respectively, in our context. In both figures, we used the same model and parameter settings as for Figure \ref{fig:Epi0example1}. However, while the graph displayed in Figure \ref{fig:Epi0example1} originated from exact analytical calculations, Figures \ref{fig:varianceidandd} and \ref{fig:MSEidandd} display the results of Monte Carlo simulations with $100,000$ repetitions. The left graphs of Figures \ref{fig:varianceidandd} and \ref{fig:MSEidandd} refer to the situation in which $p_{1}^{LFC}(X),\ldots,p_{m}^{LFC}(X)$ are jointly stochastically independent random variables under $\vartheta$ (meaning that their copula under $\vartheta$ is the product copula), while the dependency structure among $p_{1}^{LFC}(X),\ldots,p_{m}^{LFC}(X)$ under $\vartheta$ is given by the Gumbel-Hougaard copula with copula parameter $\nu=2$ in the right graphs of Figures \ref{fig:varianceidandd} and \ref{fig:MSEidandd}. Variance (Figure \ref{fig:varianceidandd}) and MSE (Figure \ref{fig:MSEidandd}) of $\hat{\pi}_{0}(1/2,c)$ are displayed as a function of $c$, for $c=0,0.05,\ldots,1$. 

In the left graph of Figure \ref{fig:varianceidandd}, the variance of $\hat{\pi}_{0}(1/2,c)$ is decreasing in $c$. This can be explained by the fact, that in the case of jointly stochastically independent LFC-based $p$-values, any randomization (i.\ e., any choice of $c < 1$) means that additional random components contributed by $U_1, \hdots, U_m$ enter the variance of $\hat{\pi}_{0}(1/2,c)$. However, as the scaling of the vertical axis in the left graph of Figure \ref{fig:varianceidandd} reveals, this (increased) variance is in essentially all considered cases smaller than the squared bias of $\hat{\pi}_{0}(1/2,c)$; cf.\ Figure \ref{fig:Epi0example1}. So, we may conclude here that taking into account $U_1, \hdots, U_m$ increases the variance of $\hat{\pi}_{0}$, but only to a magnitude which is in essentially all considered cases smaller than that of the bias reduction achieved by randomization. This is also in line with the findings of \cite{dickhaus2013randomized}; see the discussion around Table 2 in that paper.  

In the right graph of Figure \ref{fig:varianceidandd}, the behavior of the variance of $\hat{\pi}_{0}(1/2,c)$ 
is different. Here, the randomization reduces the variance of $\hat{\pi}_{0}$, often by a considerable amount. This can be explained by the fact, that in the dependency structure among $p_{1}^{rand}(X,U_{1},c), \hdots, p_{m}^{rand}(X,U_{m},c)$ the Gumbel-Hougaard copula of $p_{1}^{LFC}(X),\ldots,p_{m}^{LFC}(X)$ and the product copula of $U_1, \hdots, U_m$ are "mixed", meaning that the degree of dependency among $p_{1}^{rand}(X,U_{1},c), \hdots, p_{m}^{rand}(X,U_{m},c)$ is smaller than that among $p_{1}^{LFC}(X),\ldots,p_{m}^{LFC}(X)$. 

Furthermore, comparing the scalings of the vertical axes in the two graphs of  Figure \ref{fig:varianceidandd}, we can confirm the previous findings by \cite{neumann2017estimating} (and other authors), that (positively) dependent $p$-values lead to an increased variance of  $\hat{\pi}_{0}$ when compared with the case of jointly stochastically independent $p$-values. As can be seen from the representation on the right-hand side of \eqref{eq:schwederspjotvoll}, the variance of $\hat{\pi}_{0}$ is  essentially a re-scaled version of the variance of  $\hat{F}_{m}$. Finally, Figure \ref{fig:MSEidandd} demonstrates that the (squared) bias is the dominating part in the bias-variance decomposition of $\hat{\pi}_{0}(1/2,c)$. In particular, the shapes of the curves in 
Figure \ref{fig:MSEidandd} closely resemble the shape of the curve in Figure \ref{fig:Epi0example1}. Our conclusion is, that choosing $c = c^{\star}$ does not only minimize the bias of $\hat{\pi}_{0}(1/2,c)$, but also leads to a small MSE of $\hat{\pi}_{0}(1/2,c)$.


\subsection{Estimating $\pi_{0}$ in practice}\label{sec:practice}

The expected value in $h_{\vartheta}(\lambda,c)=\mathbb{E}_{\vartheta}[\hat{\pi}_{0}(\lambda,c)]$ discussed in Section \ref{sec:pi0theory} refers to the joint distribution of $\{U_{j}\}_{1 \leq j \leq m}$ and the data $X$ under $\vartheta$. In practice, the distribution of $X$ under $\vartheta$ is unknown, but we have a realized data sample $X = x\in\Omega$ at hand, from which $p_{1}^{LFC}(x),\ldots,p_{m}^{LFC}(x)$ can be computed. Throughout this section, let us assume a statistical model such that any of the conditions $(2.)$ -- $(5.)$ from Theorem \ref{thm:randpvalidity} is fulfilled, so that $p_{1}^{rand}(X,U_{1},c),\ldots,$ $p_{m}^{rand}(X,U_{m},c)$ are valid $p$-values for any $c\in[0,1]$.

In analogy to \eqref{eq:EFm}, we obtain that the conditional expected value (with respect to the $U_{j}$'s) of $\hat{\pi}_{0}(\lambda,c)$ under the condition $X = x$ is given by
\begin{equation}\label{eq:expectedvaluepi0hat}
\mathbb{E}[\hat{\pi}_{0}(\lambda,c)\mid X=x]=\frac{1}{1-\lambda}\left[1-\frac{1}{m}\sum_{j=1}^{m}\Big[\lambda\,\textbf{1}\{p_{j}^{LFC}(x)\geq c\}+\textbf{1}\{p_{j}^{LFC}(x)\leq\lambda c\}\Big]\right].
\end{equation}
Our proposal for practical purposes is to minimize \eqref{eq:expectedvaluepi0hat} with respect to $c \in [0, 1]$, for fixed $\lambda \in [0, 1)$. Denoting the solution of this minimization problem by $c_0$, we then propose to utilize $p_{1}^{rand}(x,U_{1},c_0),\ldots,p_{m}^{rand}(x,U_{m},c_0)$ in $\hat{\pi}_0(\lambda)$.

Minimizing \eqref{eq:expectedvaluepi0hat} with respect to $c \in [0, 1]$ is equivalent to maximizing the function $c \mapsto g_x(\lambda, c)$, given by
\begin{equation}
\label{eq:glambda}
g_x(\lambda, c) =\sum_{j=1}^{m}\big(\lambda\,\textbf{1}\{p_{j}^{LFC}(x)\geq c\}+\textbf{1}\{p_{j}^{LFC}(x)\leq\lambda c\}\big),
\end{equation}
with respect to $c \in [0, 1]$. Hence, the solution $c_{0}$ is such, that most of the (realized) LFC-based $p$-values are outside of the interval $(\lambda c_{0},c_{0})$. An optimal choice $c_{0}$ can be determined numerically by either evaluating $g_x(\lambda, \cdot)$ on a given grid $0=c_{0}<\cdots<c_{N}=1$ or on the set $\{p_{1}^{LFC}(x),\ldots,p_{m}^{LFC}(x),$ $p_{1}^{LFC}(x)/\lambda,\ldots,p_{m}^{LFC}(x)/\lambda\}$ (excluding values larger than $1$). Notice, that $g_x(\lambda, \cdot)$ can only change its values at points from the second set.

We demonstrate this procedure with an example. Again, consider the multiple $Z$-tests model and the same parameter setting as for deriving the left graph in Figure \ref{fig:varianceidandd}. Under these settings, we randomly drew one sample $x\in\Omega$ and applied the proposed procedure with $\lambda=1/2$. After the removal of elements exceeding one from the set $\{p_{1}^{LFC}(x),\ldots,p_{m}^{LFC}(x), 2 p_{1}^{LFC}(x),\ldots, 2 p_{m}^{LFC}(x)\}$, $1,406$ relevant points remained for the evaluation of $g_x(1/2, \cdot)$. As displayed in 
Figure \ref{fig:gplot}, the maximum of $g_x(1/2, \cdot)$ is for the observed $x$ attained at $c_{0} = 0.3286$. This is an optimal $c$ given the realized values $p_{1}^{LFC}(x),\ldots,p_{m}^{LFC}(x)$. For comparison, recall that we have seen in Section \ref{sec:pi0theory} that $c^{\star}=0.3276$ minimizes the bias of $\hat{\pi}_{0}(1/2, c)$ on average over $X \sim \mathbb{P}_\vartheta$.

Figure \ref{fig:ecdfswithcstar} displays the ecdfs pertaining to $p_{1}^{LFC}(x),\ldots,p_{m}^{LFC}(x)$ and\\ $p_{1}^{rand}(x,u_{1},c_{0}),\ldots,p_{m}^{rand}(x,u_{m},c_{0})$, respectively, where $\{u_{1},\ldots,u_{m}\}$
is one particular set of realizations of the random variables $U_{1},\ldots,U_{m}$. 
Furthermore, the two dotted vertical lines in Figure \ref{fig:ecdfswithcstar} indicate the interval $[c_0/2, c_0]$. Recall that $c_0$ is chosen such, that most of the (realized) LFC-based $p$-values are outside of the latter interval. This can visually be confirmed, since the ecdf pertaining to $p_{1}^{LFC}(x),\ldots,p_{m}^{LFC}(x)$ is rather flat on $[c_0/2, c_0]$.

For any ecdf $t \mapsto \hat{F}_m(t)$ utilized in $\hat{\pi}_0(\lambda)$, the offset at $t=0$ of the straight line connecting the points $(1,1)$ and $(\lambda,\hat{F}_{m}(\lambda))$ equals $1 -\hat{\pi}_0(\lambda)$; cf., e.\ g., Figure 3.2.(b) in \cite{Dickhaus-Buch2014}.
We therefore obtain an accurate estimate of $\pi_0$ if the ecdf $t \mapsto \hat{F}_{m}(t)$ utilized in $\hat{\pi}_0(\lambda)$ is at $t=\lambda$ close to the straight line connecting the points $(1,1)$ and $(0,1-\pi_{0})$. The latter "optimal" line is the expected ecdf of marginal $p$-values that are $\mathrm{Uni}[0,1]$-distributed under the null and almost surely equal to zero under the alternative. In Figure \ref{fig:ecdfswithcstar}, the ecdf pertaining to $p_{1}^{rand}(x,u_{1},c_{0}),\ldots,p_{m}^{rand}(x,u_{m},c_{0})$ is much closer to that optimal line than the ecdf pertaining to $p_{1}^{LFC}(x),\ldots,p_{m}^{LFC}(x)$. Consequently, for this particular dataset the estimation approach based on $p_{1}^{rand}(x,u_{1},c_{0}),\ldots,p_{m}^{rand}(x,u_{m},c_{0})$ leads to a much more precise estimate of $\pi_0 = 0.7$ then the one based on 
$p_{1}^{LFC}(x),\ldots,p_{m}^{LFC}(x)$. The estimate based on $p_{1}^{LFC}(x),\ldots,p_{m}^{LFC}(x)$ even exceeds one in this example. We have repeated this simulation several times (results not included here) and the conclusions have always been rather similar.

\section{Discussion}
\label{sec:discussion}
We have demonstrated how randomized $p$-values can be utilized in the Schweder-Spj\o tvoll estimator $\hat{\pi}_0$. Whenever composite null hypotheses are under consideration, our proposed approach leads to a reduction of the bias and of the MSE of $\hat{\pi}_0$, when compared to the usage of LFC-based $p$-values. Furthermore, our approach also robustifies $\hat{\pi}_0$ against dependencies among $p_{1}^{LFC}(X),\ldots,$ $p_{m}^{LFC}(X)$. The latter property is important in modern high-dimensional applications, where the biological and/or technological mechanisms involved in the data-generating process virtually always lead to dependencies (cf. \cite{MADAM}), especially in studies with multiple endpoints which are all measured for the same observational units. Furthermore, we have explained in detail how the proposed methodology can be applied in practice. Worksheets in \verb=R=, with which all results of the present work can be reproduced, are available from the first author upon request. 

Statistical models that fulfill any of the conditions $(2.)$ -- $(5.)$ from Theorem \ref{thm:randpvalidity} admit valid randomized $p$-values $\{p_{j}^{rand}(X,U_{j},c_{j})\}_{1 \leq j \leq m}$ for any choice of the constants $(c_{j})_{1 \leq j \leq m} \in[0,1]^m$. We gave two such models in Examples \ref{ex:multipleztests} and \ref{ex:twosamples}. These models have a variety of applications, for instance in the life sciences; cf., e.\ g., Part II of \cite{Dickhaus-Buch2014}. Closely related examples are the replicability models considered in \cite{hoang2019randomized}. Identifying additional model classes that have that property can be addressed in future research. Furthermore, in models for which the $j$-th LFC-based $p$-value is of the form  $p_{j}^{LFC}(X)=1-F_{\vartheta_{0}}(T_{j}(X))$ for $1 \leq j \leq m$ and in which $(T_{j}(X)^{(\vartheta)})_{\theta_{j}(\vartheta)}$ is an MLR family, the cdf of $p_{j}^{rand}(X,U_{j},R_{j})$ is always between those of $\mathrm{Uni}[0,1]$ and $p_{j}^{LFC}(X)$. Distributions with the MLR property include exponential families, for example the family of univariate normal distributions with fixed variance and the family of Gamma distributions (cf. \cite{karlin1956distributions}). Also, the family of non-central $t$-distributions and the family of non-central 
$F$-distributions have the MLR property with respect to their non-centrality parameters (cf. \cite{karlin1956decision}). It is of interest to deeper investigate properties of our randomized $p$-values in such models.

There are several further possible extensions of the present work. First, in Section \ref{sec:proportionoftruenullhypotheses} we only considered the usage of $p_{1}^{rand}(X,U_{1},c_{1}),\ldots,p_{m}^{rand}(X,U_{m},c_{m})$ in $\hat{\pi}_{0}$ for identical constants $c_{1}=\cdots=c_{m} \equiv c$. In future work, it may be of interest to develop a method for choosing each $c_{j}$ individually, for instance depending on the size of the $j$-th LFC-based $p$-value. 
Second, we have chosen $c_0$ in Section \ref{sec:practice} such, that the conditional (to the observed data $X = x$) bias of $\hat{\pi}_{0}(\lambda)$ is minimized. Another approach, which can be pursued in future research, is to choose a $c_{0}$ that minimizes the MSE of $\hat{\pi}_{0}(\lambda)$ instead.
Third, we restricted our attention to the Schweder-Spj\o tvoll estimator $\hat{\pi}_{0}(\lambda)$. However, there exists a wide variety of other ecdf-based estimators in the literature (see, for instance, Table 1 in \cite{Chen2019} for a recent overview), which are prone to suffer from the same issues as $\hat{\pi}_{0}(\lambda)$ when used with LFC-based $p$-values in the context of composite null hypotheses. One other ecdf-based estimator for $\pi_{0}$ is the more conservative estimator $\hat{\pi}_{0}^{+}(\lambda)=\hat{\pi}_{0}(\lambda)+1/(m(1-\lambda))$ proposed by \cite{storey2002direct}. The bias of $\hat{\pi}_{0}^{+}$ when used with the randomized $p$-values $p_{1}^{rand}(X,U_{1},c),\ldots,p_{m}^{rand}(X,U_{m},c)$ is minimized for the same $c=c^{\star}$ from Section \ref{sec:proportionoftruenullhypotheses}. Thus, the same algorithm as outlined in Section \ref{sec:practice} can be applied to $\hat{\pi}_{0}^{+}$ in practice. In future research, randomization approaches for other ecdf-based estimators can be investigated.
Finally, we have not elaborated on the choice of $\lambda$ in the present work. The standard choice of $\lambda = 1/2$ seemed to work reasonably well in connection with our proposed randomized $p$-values. We have also performed some preliminary sensitivity analyses (not included here) with respect to $\lambda$, which indicated that the sensitivity of $\hat{\pi}_{0}$ with respect to $\lambda$ is less pronounced for the case of randomized  $p$-values than for the case of LFC-based $p$-values. Investigating this phenomenon deeper, both from the theoretical and from the numerical perspective,  is also a worthwhile topic for future research.


\section*{Appendix}
\label{sec:appendix}
\subsection*{The more general randomized $p$-values}
\label{sec:generaldefinitionprand}
\subsubsection*{Definition}
\label{def:generalprand}

Let $U_{1},\ldots,U_{m}$ and $X$ be as before. For a set of stochastically independent (not necessarily identically distributed) random variables $R_{1},\ldots,R_{m}$ with values in $[0,1]$, that are defined on the same probability space as $X$, stochastically independent of the $U_{j}$'s and the data $X$, and whose distributions do not depend on $\vartheta$, we define 
\begin{equation}\label{doublerand}
p_{j}^{rand}(X,U_{j},R_{j})=U_{j} \mathbf{1}\{p_{j}^{LFC}(X)\geq R_{j}\}+\frac{p_{j}^{LFC}(X)}{R_{j}}\mathbf{1}\{p_{j}^{LFC}(X)<R_{j}\},
\end{equation}
$j=1,\ldots,m$.
This definition includes the case $R_{j}\equiv c_{j}$ from Definition$~\ref{def:randomizedpvalues}$ for any constant $c_{j}\in[0,1]$, $j=1,\ldots,m$. We generalize and prove Theorems$~\ref{thm:randpvalidity}$ and $\ref{thm:stochasticorder}$ for the randomized $p$-values $\{p_{j}^{rand}(X,U_{j},R_{j})\}_{1 \leq j \leq m}$.

\subsubsection*{Theorem~\ref{thm:randpvalidity}\,$'$}
\label{thm:generalthmrandpvalidity}

Let a model as in Section$~\ref{sec:modelsetup}$ be given and $j\in\{1,\ldots,m\}$ be fixed. Then, the $j$-th randomized $p$-value $p_{j}^{rand}(X,U_{j},R_{j})$ as in \eqref{doublerand} is a valid $p$-value for a given random variable $R_{j}$ with values in $[0,1]$ if and only if condition $(0.)$ is fulfilled. Furthermore, either of the following conditions $(1.')$, $(2.)$, and $(3.)$ is a sufficient condition for the validity of $p_{j}^{rand}(X,U_{j},R_{j})$ for any random variable $R_{j}$ with values in $[0,1]$.
\begin{description}

\item[$(0.)$] For every $\vartheta\in\Theta$ with $\theta_{j}(\vartheta)\in H_{j}$, it holds
\begin{equation*}
\mathbb{P}_{\vartheta}(p_{j}^{LFC}(X)\leq t R_{j})\leq t \mathbb{P}_{\vartheta}(p_{j}^{LFC}(X)\leq R_{j})
\end{equation*}
for all $t\in[0,1]$.

\item[$(1.')$] For every $\vartheta\in\Theta$ with $\theta_{j}(\vartheta)\in H_{j}$, it holds
\begin{equation*}
\mathbb{P}_{\vartheta}(p_{j}^{LFC}(X)\leq t u)\leq t \mathbb{P}_{\vartheta}(p_{j}^{LFC}(X)\leq u)
\end{equation*}
for all $u,t\in[0,1]$.

\item[$(2.)$] For every $\vartheta\in\Theta$ with $\theta_{j}(\vartheta)\in H_{j}$, $\mathbb{P}_{\vartheta}(p_{j}^{LFC}(X)\leq t)/t$ is non-decreasing in $t$.

\item[$(3.)$] The cdf of $p_{j}^{LFC}(X)$ is convex under any 
$\vartheta\in\Theta$ with $\theta_{j}(\vartheta)\in H_{j}$.

\end{description}

Let $F_{\vartheta}$ be the cdf of $T_{j}(X)$ under $\vartheta\in\Theta$. If the LFC-based $p$-value is given by $p_{j}^{LFC}(X)=1-F_{\vartheta_{0}}(T_{j}(X))$, where $\vartheta_{0}\in\Theta$ is an LFC for $\varphi_{j}$, then the following condition $(4.)$ is equivalent to condition $(2.)$, while condition $(5.)$ is equivalent to condition $(3.)$.

\begin{description}

\item[$(4.)$] For every $\vartheta\in\Theta$ with $\theta_{j}(\vartheta)\in H_{j}$, it holds $T_{j}(X)^{(\vartheta)}\leq_{\mathrm{hr}}T_{j}(X)^{(\vartheta_{0})}$.

\item[$(5.)$] For every $\vartheta\in\Theta$ with $\theta_{j}(\vartheta)\in H_{j}$, it holds $T_{j}(X)^{(\vartheta)}\leq_{\mathrm{lr}}T_{j}(X)^{(\vartheta_{0})}$.

\end{description}

\begin{proof}
First, we show that condition $(0.)$ is equivalent to $p_{j}^{rand}(X,U_{j},R_{j})$ being valid. For $p_{j}^{rand}(X,U_{j},R_{j})$ to be valid it has to hold that
\begin{equation*}
\mathbb{P}_{\vartheta}(p_{j}^{rand}(X,U_{j},R_{j})\leq t)\leq t,
\end{equation*}
for all $t\in[0,1]$ and all $\vartheta\in\Theta$ with $\theta_{j}(\vartheta)\in H_{j}$. It holds that
\begin{eqnarray}
\mathbb{P}_{\vartheta}(p_{j}^{rand}(X,U_{j},R_{j})\leq t)&= &\mathbb{P}_{\vartheta}(U_{j}\leq t) \mathbb{P}_{\vartheta}(p_{j}^{LFC}(X)>R_{j})+\mathbb{P}_{\vartheta}(p_{j}^{LFC}(X)\leq t R_{j})\nonumber \\ 
&= &t \mathbb{P}_{\vartheta}(p_{j}^{LFC}(X)>R_{j})+\mathbb{P}_{\vartheta}(p_{j}^{LFC}(X)\leq t R_{j}).\label{eq:thm1eq1}
\end{eqnarray}
Now, the term in \eqref{eq:thm1eq1} is not larger than $t$ if and only if it holds
\begin{equation*}
\mathbb{P}_{\vartheta}(p_{j}^{LFC}(X)\leq t R_{j})\leq t\big[1-\mathbb{P}_{\vartheta}(p_{j}^{LFC}(X)>R_{j})\big]=t \mathbb{P}_{\vartheta}(p_{j}^{LFC}(X)\leq R_{j}),
\end{equation*}
which is condition $(0.)$.

Let $G$ be the cdf of $R_{j}$. From condition $(1.')$ it follows
\begin{equation*}
\int_{0}^{1}\mathbb{P}_{\vartheta}(p_{j}^{LFC}(X)\leq t u) \mathrm{d} G(u)\leq t \int_{0}^{1}\mathbb{P}_{\vartheta}(p_{j}^{LFC}(X)\leq u) \mathrm{d} G(u),
\end{equation*}
for every $t\in[0,1]$, thus condition $(1.')$ implies $(0.)$.

Substituting $z=t u$ in condition $(1.')$ leads to 
\begin{equation*}
\mathbb{P}_{\vartheta}(p_{j}^{LFC}(X)\leq z)\leq z\frac{\mathbb{P}_{\vartheta}(p_{j}^{LFC}(X)\leq u)}{u}
\end{equation*}
for all $0\leq z<u\leq 1$ and all $\vartheta\in\Theta$ with $\theta_{j}(\vartheta)\in H_{j}$, which is equivalent to condition $(2.)$.

Now, we show that condition $(3.)$ implies condition $(1.')$. Let $u\in[0,1]$ be fixed. The inequality in $(1.')$ is always satisfied for $t=0$ and $t=1$. Since $t\mapsto t \mathbb{P}_{\vartheta}(p_{j}^{LFC}(X)\leq u)$ is a linear function and $t\mapsto\mathbb{P}_{\vartheta}(p_{j}^{LFC}(X)\leq t u)$ is a convex function, if $(3.)$ is fulfilled, it holds 
\begin{equation*}
\mathbb{P}_{\vartheta}(p_{j}^{LFC}(X)\leq t u)\leq t \mathbb{P}_{\vartheta}(p_{j}^{LFC}(X)\leq u)
\end{equation*}  
for all $t\in[0,1]$.

Now we assume that $p_{j}^{LFC}(X)=1-F_{\vartheta_{0}}(T_{j}(X))$. At first we show that conditions $(2.)$ and $(4.)$ are equivalent. To this end, notice that the term
\begin{equation*}
\frac{\mathbb{P}_{\vartheta}(p_{j}^{LFC}(X)\leq t)}{t}=\frac{\mathbb{P}_{\vartheta}(p_{j}^{LFC}(X)\leq t)}{\mathbb{P}_{\vartheta_{0}}(p_{j}^{LFC}(X)\leq t)}=\frac{\mathbb{P}_{\vartheta}(T_{j}(X)\geq F_{\vartheta_{0}}^{-1}(1-t))}{\mathbb{P}_{\vartheta_{0}}(T_{j}(X)\geq F_{\vartheta_{0}}^{-1}(1-t))}
\end{equation*}
is non-decreasing in $t$ if and only if $\mathbb{P}_{\vartheta}(T_{j}(X)\geq z)/\mathbb{P}_{\vartheta_{0}}(T_{j}(X)\geq z)=(1-F_{\vartheta}(z))/(1-F_{\vartheta_{0}}(z))$ is non-increasing in $z$.

Lastly, we show that conditions $(3.)$ and $(5.)$ are equivalent. Let $f_{\vartheta}$ be the Lebesgue density of $T_{j}(X)$ under $\vartheta\in\Theta$. Let $\vartheta\in\Theta$ be such that $\theta_{j}(\vartheta)\in H_{j}$ holds. The convexity of $t\mapsto\mathbb{P}_{\vartheta}(p_{j}^{LFC}(X)\leq t)$ is equivalent to 
\begin{align*}
\frac{\mathrm{d}}{\mathrm{d}t}\mathbb{P}_{\vartheta}(T_{j}(X)\geq F_{\vartheta_{0}}^{-1}(1-t))&=\frac{\mathrm{d}}{\mathrm{d}t}\big[1-F_{\vartheta}(F_{\vartheta_{0}}^{-1}(1-t))\big]\\
&=\frac{f_{\vartheta}(F_{\vartheta_{0}}^{-1}(1-t))}{f_{\vartheta_{0}}(F_{\vartheta_{0}}^{-1}(1-t))}
\end{align*}
being non-decreasing in $t$, or $f_{\vartheta}(z)/f_{\vartheta_{0}}(z)$ being non-increasing in $z$, which is equivalent to condition $(5.)$; cf. the remarks after Theorem$~\ref{thm:randpvalidity}$.
\qed
\end{proof}

In Theorem$~\ref{thm:randpvalidity}'$, the conditions $(2.)$ -- $(5.)$ are the same as in Theorem$~\ref{thm:randpvalidity}$. Condition $(1.')$ is equivalent to condition $(1.)$ in Theorem$~\ref{thm:randpvalidity}$ holding for all $c_{j}\in[0,1]$. Thus, $p_{j}^{rand}(X,U_{j},c_{j})$ being valid for all $c_{j}\in[0,1]$ implies the validity of $p_{j}^{rand}(X,U_{j},R_{j})$ for any random variable $R_{j}$ on $[0,1]$, $j=1,\ldots,m$. The reverse is also true, thus, the randomized $p$-value $p_{j}^{rand}(X,U_{j},R_{j})$ is valid for any random variable $R_{j}$ on $[0,1]$ if and only if it is valid for $R_{j}\equiv c_{j}$, for all $c_{j}\in[0,1]$, $j=1,\ldots,m$.

In the following, we show that Theorem$~\ref{thm:stochasticorder}$ still holds if we replace the constants $c_{j}\leq\tilde{c}_{j}$ by the random variables $R_{j}\leq_{\mathrm{st}}\tilde{R}_{j}$.

\subsubsection*{Theorem~\ref{thm:stochasticorder}\,$'$}
\label{thm:generalthmstochasticorder}

Let a model as in Section$~\ref{sec:modelsetup}$ be given and $j\in\{1,\ldots,m\}$ be fixed. If the cdf of $p_{j}^{LFC}(X)$ is convex under a fixed $\vartheta\in\Theta$, then it is
\begin{equation*}
p_{j}^{rand}(X,U_{j},R_{j})^{(\vartheta)}\leq_{\mathrm{st}}p_{j}^{rand}(X,U_{j},\tilde{R}_{j})^{(\vartheta)}
\end{equation*}
for any random variables $R_{j},\tilde{R}_{j}$ on $[0,1]$, with $R_{j}\leq_{\mathrm{st}}\tilde{R}_{j}$.

If the cdf of $p_{j}^{LFC}(X)$ is concave under a fixed $\vartheta\in\Theta$, then it holds that
\begin{equation*}
p_{j}^{rand}(X,U_{j},\tilde{R}_{j})^{(\vartheta)}\leq_{\mathrm{st}}p_{j}^{rand}(X,U_{j},R_{j})^{(\vartheta)}
\end{equation*}
for any random variables $R_{j}$ and $\tilde{R}_{j}$ with values in $[0,1]$ and with $R_{j}\leq_{\mathrm{st}}\tilde{R}_{j}$.

\begin{proof}
We first show both statements in Theorem$~\ref{thm:stochasticorder}'$ for constants $0\leq c_{j}\leq \tilde{c}_{j}\leq 1$ instead of random variables $R_{j}$ and $\tilde{R}_{j}$, which amounts to the statements in Theorem$~\ref{thm:stochasticorder}$.

For every fixed $t\in[0,1]$ and fixed $\vartheta\in\Theta$ we define the function
$q: [0, 1] \to [0, 1]$ by 
\begin{equation*}
q(c)=\mathbb{P}_{\vartheta}(p_{j}^{rand}(X,U_{j},c)\leq t)=t \mathbb{P}_{\vartheta}(p_{j}^{LFC}(X)>c)+\mathbb{P}_{\vartheta}(p_{j}^{LFC}(X)\leq c t).
\end{equation*}  
Furthermore, we denote by $f_{\vartheta}$ the Lebesgue density of $p_{j}^{LFC}(X)$ under $\vartheta$, such that it holds $q'(c)=-t f_{\vartheta}(c)+t f_{\vartheta}(c t)$, which is not positive if $f_{\vartheta}$ is non-decreasing and not negative if $f_{\vartheta}$ is non-increasing.

Let $R_{j}$ and $\tilde{R}_{j}$ be random variables fulfilling the assumptions of the theorem. 
If $q$ is non-decreasing, then it holds that $\mathbb{E}[q(R_{j})]\leq\mathbb{E}[q(\tilde{R}_{j})]$, and if $q$ is non-increasing it holds that $\mathbb{E}[q(R_{j})]\geq\mathbb{E}[q(\tilde{R}_{j})]$, where $\mathbb{E}$ refers to the joint distribution of $R_j$ and $\tilde{R}_{j}$. Since $\mathbb{E}[q(R_{j})]=\mathbb{P}_{\vartheta}(p_{j}^{rand}(X,U_{j},R_{j})\leq t)$ and $\mathbb{E}[q(\tilde{R}_{j})]=\mathbb{P}_{\vartheta}(p_{j}^{rand}(X,U_{j},\tilde{R}_{j})\leq t)$, the proof is completed.
\qed
\end{proof}

\begin{acknowledgements}
Financial support by the Deutsche Forschungsgemeinschaft via grant No. DI 1723/5-1 is gratefully acknowledged.
\end{acknowledgements}

\bibliographystyle{spbasic}      
\bibliography{Randomized-Schweder-arXiv}   


\newpage
\begin{figure}
\centering
  \includegraphics[width=0.9\textwidth]{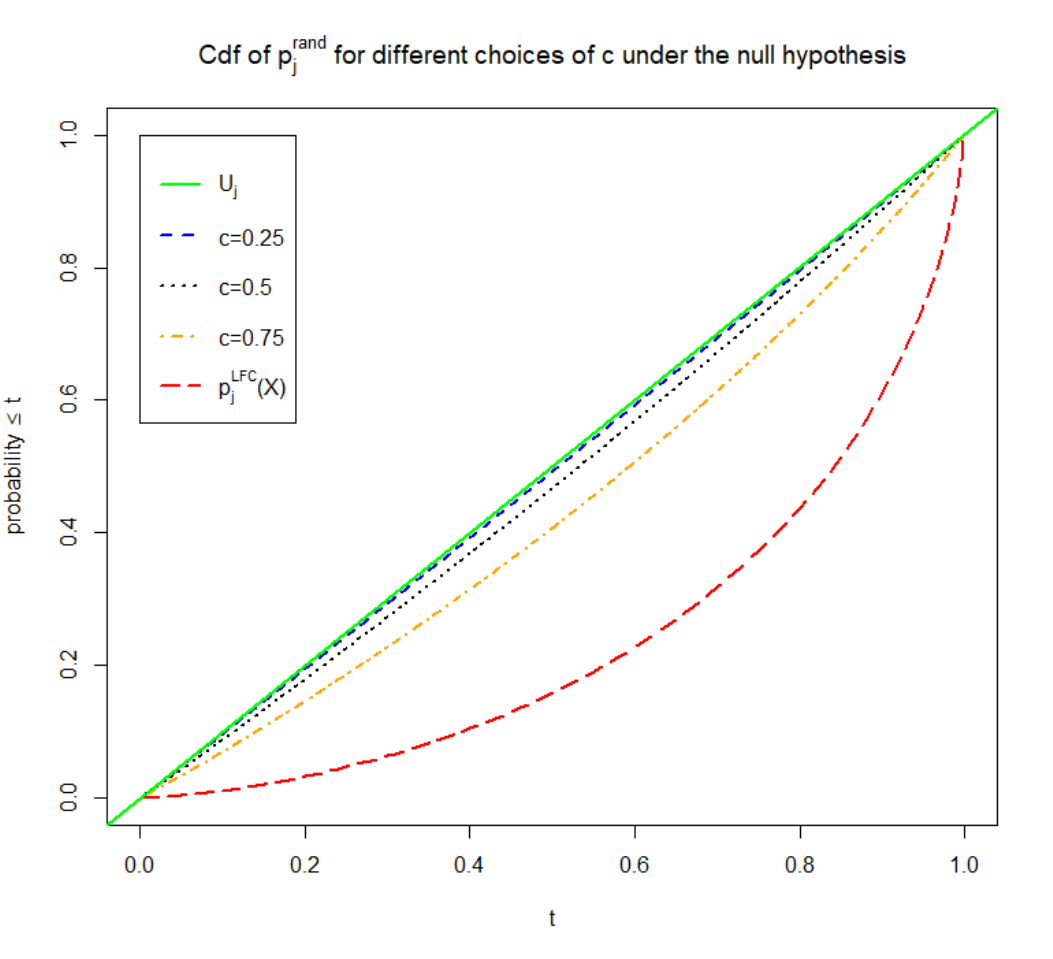}
  \caption{A comparison of the cdfs of $p_{j}^{rand}(X,U_{j},c)$, $c=0,0.25,0.5,0.75,1$, under the multiple $Z$-tests model for $\theta_{j}(\vartheta)=-1/\sqrt{n_{j}}$, where $n_{j}=50$. The value of $j \in \{1, \hdots, m\}$ is arbitrary.} 
  \label{fig:cdfsnull}
\end{figure}  

\begin{figure}
\centering
  \includegraphics[width=0.9\textwidth]{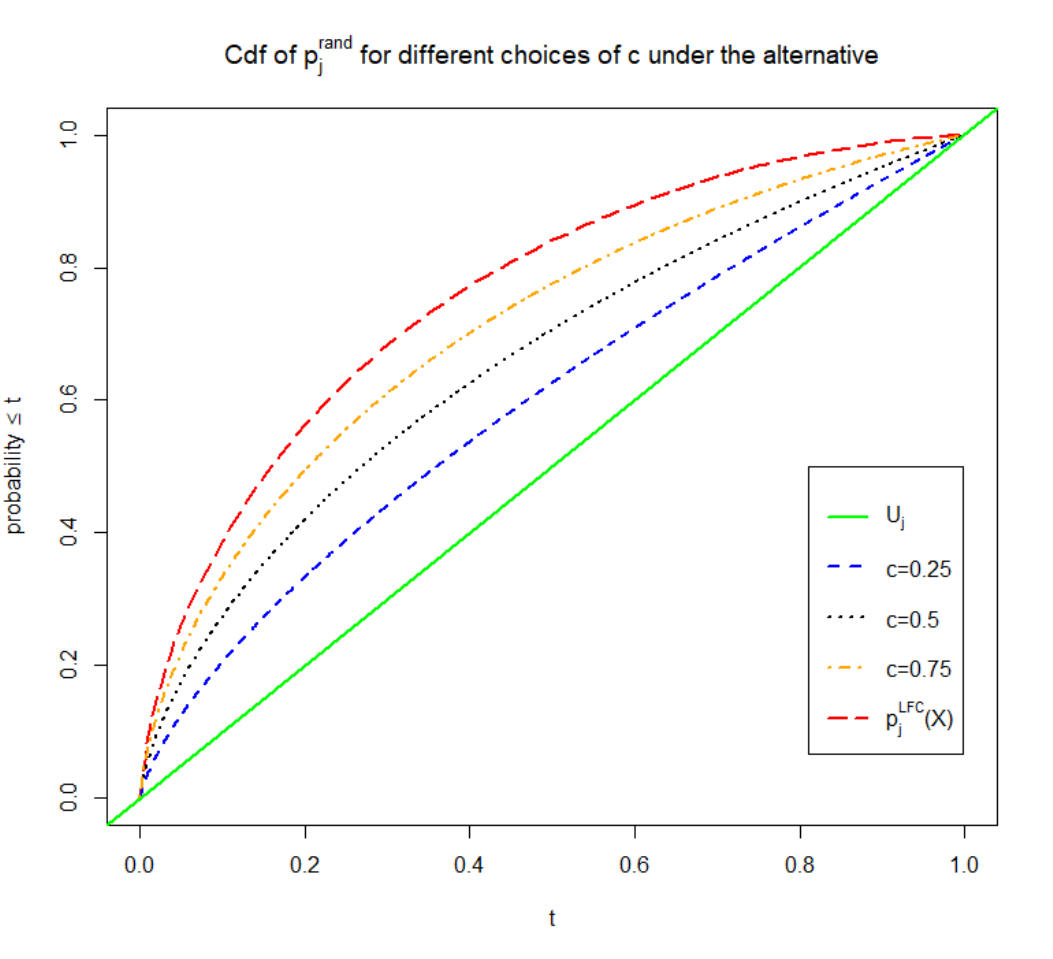}
  \caption{A comparison of the cdfs of $p_{j}^{rand}(X,U_{j},c)$, $c=0,0.25,0.5,0.75,1$, under the multiple $Z$-tests model for $\theta_{j}(\vartheta)=1/\sqrt{n_{j}}$, where $n_{j}=50$. The value of $j \in \{1, \hdots, m\}$ is arbitrary.}
  \label{fig:cdfsalternative}
\end{figure}  

\begin{figure}
\centering
  \includegraphics[width=0.9\textwidth]{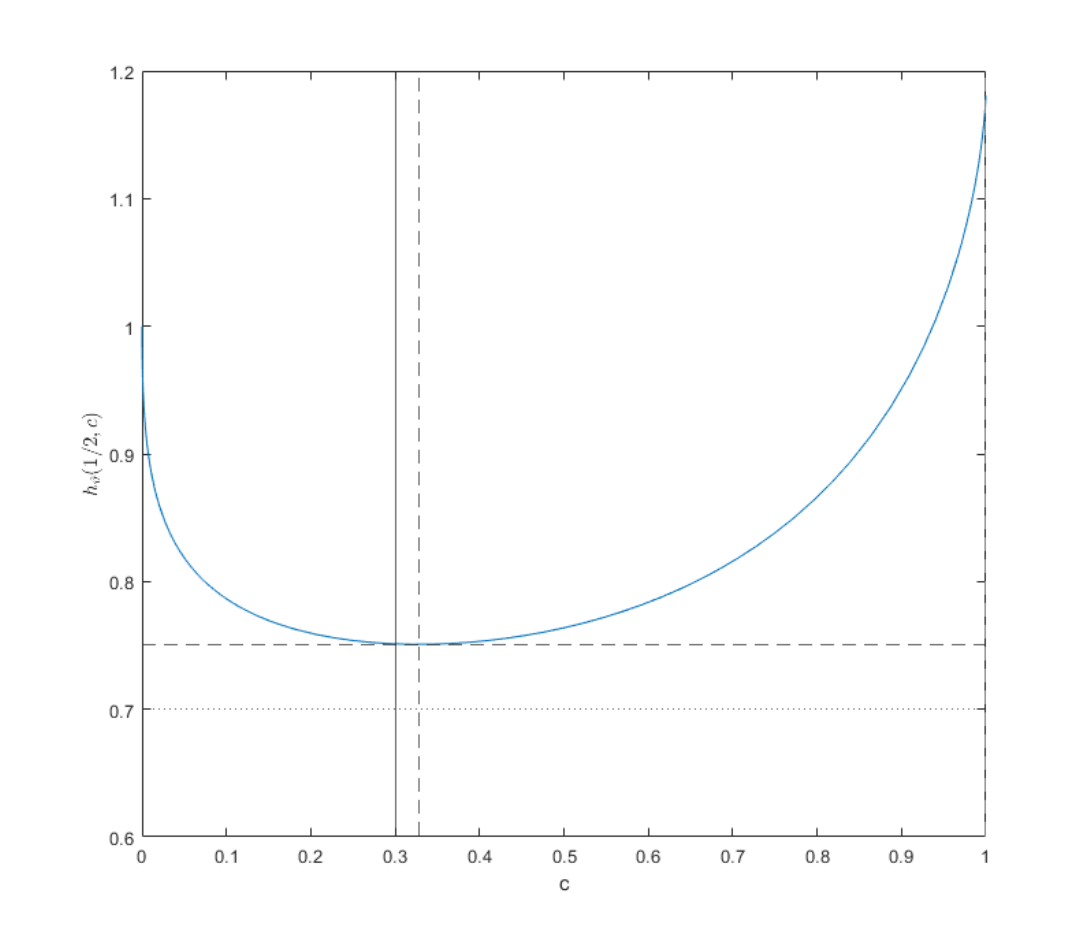}
  \caption{A plot of $c \mapsto h_{\vartheta}(1/2,c)$ for $c\in[0,1]$ under the multiple $Z$-tests model. We set $\pi_{0}=0.7$, and $\vartheta\in\Theta$ such that $\theta_{j}(\vartheta)=-1/\sqrt{50}$ if $H_{j}$ is true and $\theta_{j}(\vartheta)=2.5/\sqrt{50}$ if $K_{j}$ is true, $j=1,\ldots,m = 1{,}000$. The solid vertical line indicates $c=1-\pi_{0}$, while the dashed one indicates the minimizing argument $c^{\star}$ of $c\mapsto h_{\vartheta}(1/2,c)$. The dashed horizontal line indicates $h_{\vartheta}(1/2,c^{\star})$, while the dotted one indicates $\pi_0 = 0.7$.}
  \label{fig:Epi0example1}
\end{figure}

\begin{figure}
\centering
  \includegraphics[width=0.9\textwidth]{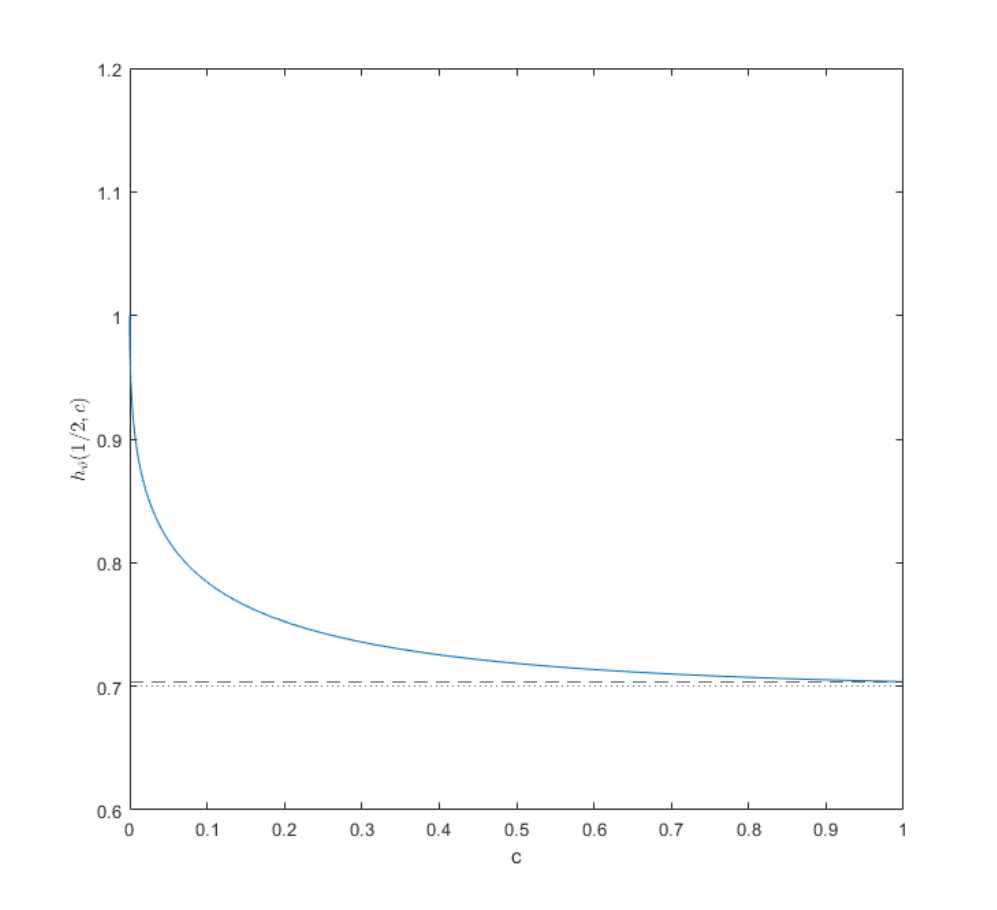}
  \caption{A plot of $c \mapsto h_{\vartheta}(1/2,c)$ for $c\in[0,1]$ under the multiple $Z$-tests model. We set $\pi_{0}=0.7$, and $\vartheta\in\Theta$ such that $\theta_{j}(\vartheta)=0$ if $H_{j}$ is true and $\theta_{j}(\vartheta)=2.5/\sqrt{50}$ if $K_{j}$ is true, $j=1,\ldots,m = 1{,}000$.}
  \label{fig:Epi0example3}
\end{figure} 

\begin{figure}
\centering
  \includegraphics[width=\textwidth]{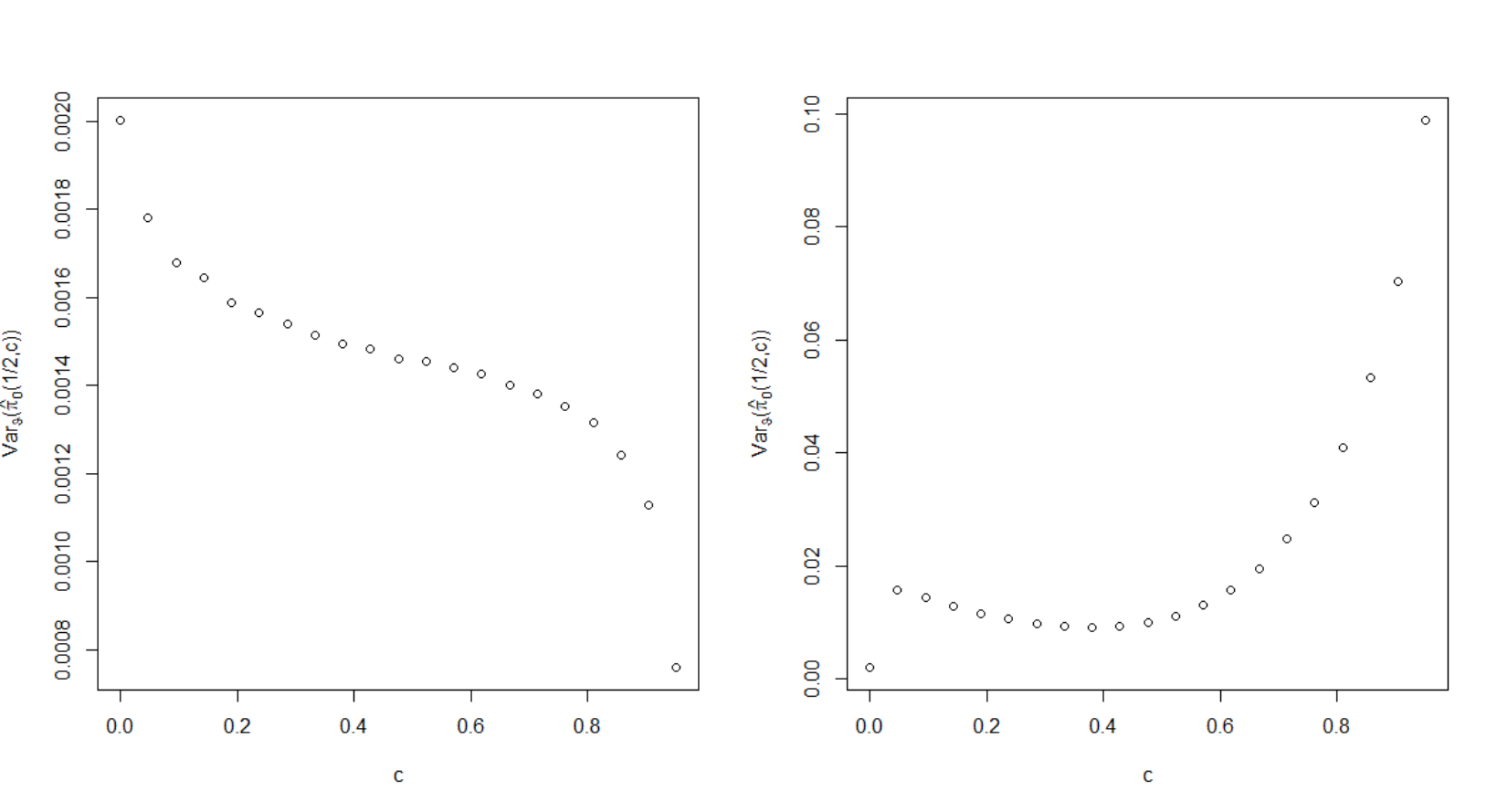}
  \caption{The variance $\mathrm{Var}_\vartheta(\hat{\pi}_{0}(1/2,c))$ for $c=0,0.05,\ldots,1$ in the multiple $Z$-tests model for $\pi_{0}=0.7$, and $\vartheta\in\Theta$ such that $\theta_{j}(\vartheta)=-1/\sqrt{50}$ if $H_{j}$ is true and $\theta_{j}(\vartheta)=2.5/\sqrt{50}$ if $K_{j}$ is true, $j=1,\ldots,m=1{,}000$. The LFC-based $p$-values are jointly stochastically independent in the left graph and have the Gumbel-Hougaard copula with copula parameter $\nu=2$ in the right graph.}
  \label{fig:varianceidandd}
\end{figure}

\begin{figure}
\centering
  \includegraphics[width=\textwidth]{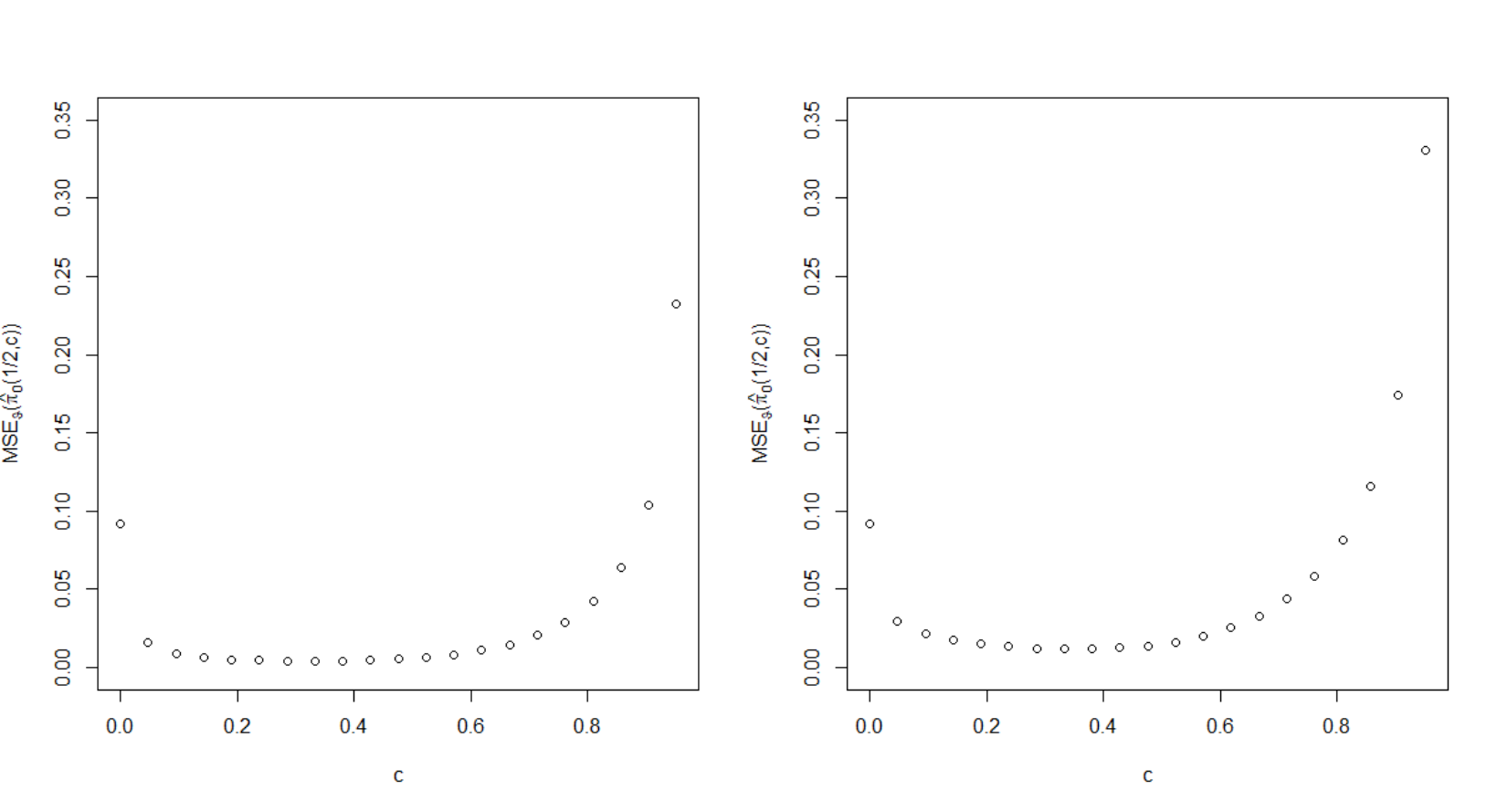}
  \caption{The mean squared error $\mathrm{MSE}_\vartheta(\hat{\pi}_{0}(1/2,c))$ for $c=0,0.05,\ldots,1$ in the multiple $Z$-tests model for $\pi_{0}=0.7$, and $\vartheta\in\Theta$ such that $\theta_{j}(\vartheta)=-1/\sqrt{50}$ if $H_{j}$ is true and $\theta_{j}(\vartheta)=2.5/\sqrt{50}$ if $K_{j}$ is true, $j=1,\ldots,m=1{,}000$. The LFC-based $p$-values are jointly stochastically independent in the left graph and have the Gumbel-Hougaard copula with copula parameter $\nu=2$ in the right graph.}
  \label{fig:MSEidandd}
\end{figure} 


\begin{figure}
\centering
  \includegraphics[width=0.9\textwidth]{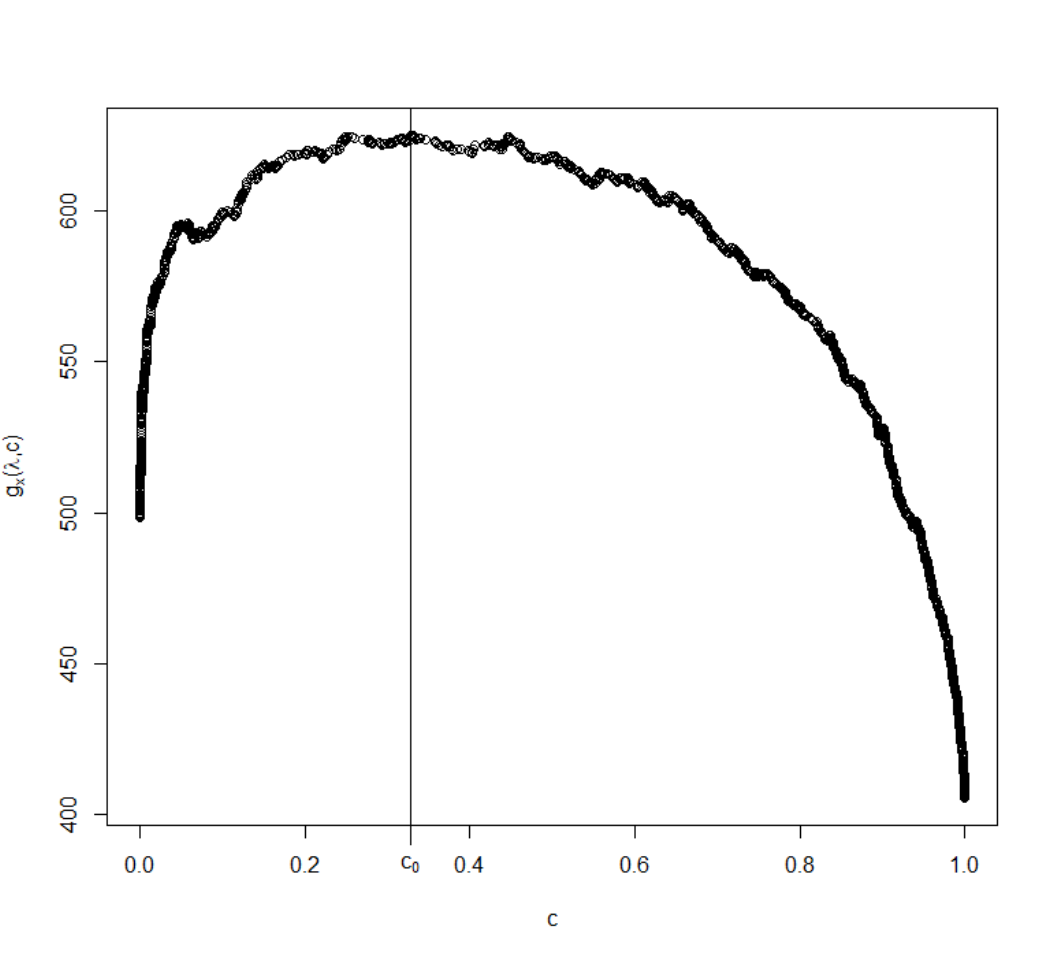}
  \caption{A plot of the function $c \mapsto g_x(\lambda, c)$, for $\lambda=1/2$, evaluated on those $1{,}406$ elements of the set $\{p_{1}^{LFC}(x),\ldots,p_{m}^{LFC}(x),\frac{p_{1}^{LFC}(x)}{\lambda},\ldots,\frac{p_{m}^{LFC}(x)}{\lambda}\}$ which are not larger than one. Here, $g_x(\lambda, \cdot)$ attains its maximum at $c_{0}=0.3286$.
	The underlying data $x$ have randomly been drawn under the multiple $Z$-tests model and the same parameter setting as for the left graph in Figure \ref{fig:varianceidandd}.}
  \label{fig:gplot}
\end{figure}  

\begin{figure}
\centering
  \includegraphics[width=0.9\textwidth]{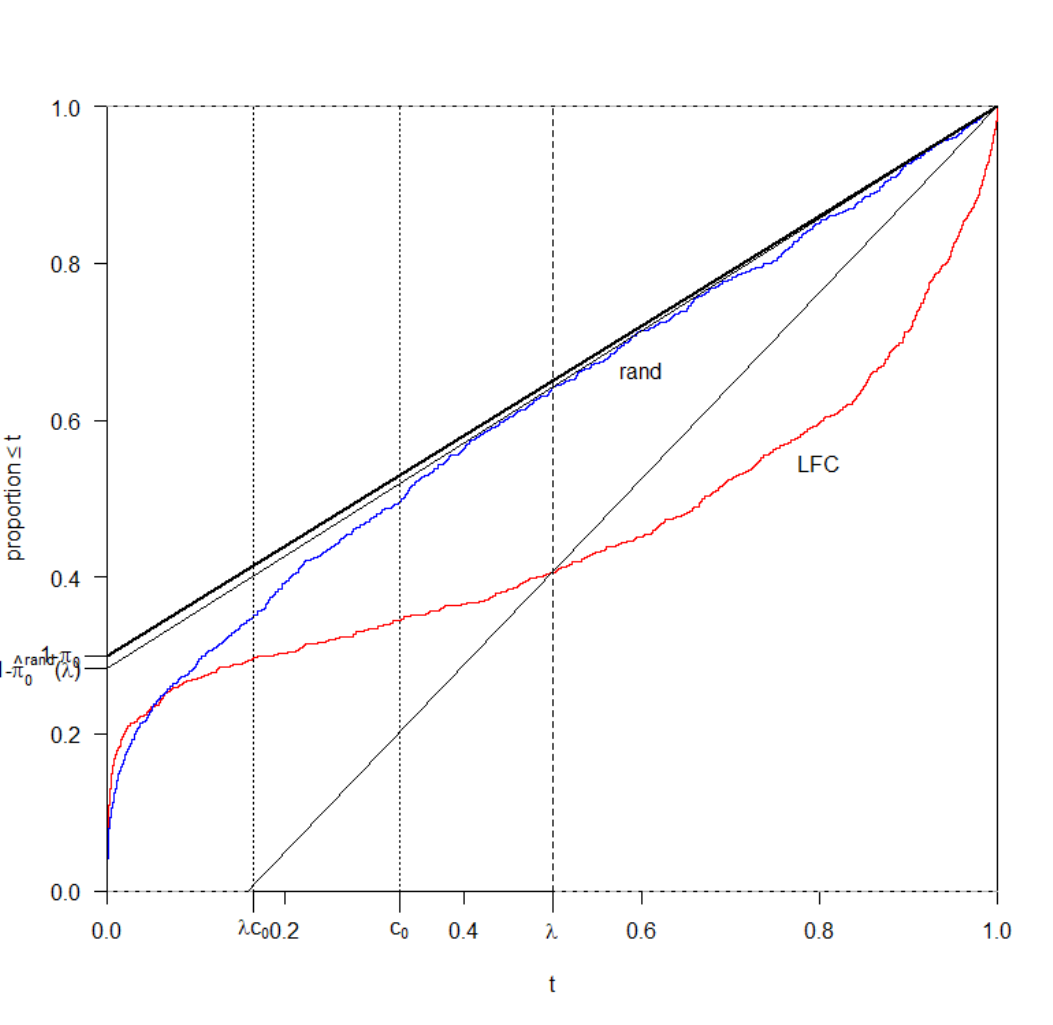}
  \caption{The ecdfs $\hat{F}_{m}$ of $(p_{j}^{LFC}(x))_{j=1,\ldots,m}$ and $(p_{j}^{rand}(x,u_{j},c_{0}))_{j=1,\ldots,m}$, respectively, under the multiple $Z$-tests model for $\pi_{0}=0.7$. The underlying data $x$ are the same as in Figure \ref{fig:gplot}. The thicker straight line connects the points $(0,1-\pi_{0})$ and $(1,1)$, while the two thinner straight lines connect $(\lambda,\hat{F}_{m}(\lambda))$ with $(1,1)$ for the two aforementioned ecdfs. The offset of each of the two thinner lines at $t=0$ equals $1-\hat{\pi}_{0}(\lambda)$ for the respective ecdf, where $\lambda=1/2$. The two dotted vertical lines indicate the interval $[\lambda c_0, c_0]$, where $c_0$ is as in Figure \ref{fig:gplot}.}
  \label{fig:ecdfswithcstar}
\end{figure}

\end{document}